\documentclass[runningheads]{llncs}

\usepackage[misc,geometry]{ifsym}  
\usepackage[utf8]{inputenc}
\usepackage{amsmath}
\usepackage{amssymb}
\usepackage{graphicx}
\usepackage{tabularx}
\usepackage{array}
\usepackage{pifont}
\usepackage{xspace}
\usepackage{todonotes}
\usepackage{paralist}
\usepackage{thm-restate}
\usepackage{subfigure}
\usepackage{authblk}
\usepackage[pdfpagelabels,colorlinks,citecolor=blue,linkcolor=blue,urlcolor=blue]{hyperref}
\usepackage[english]{babel}
\usepackage{amsopn}
\usepackage{latexsym}
\usepackage{multirow}
\usepackage{multicol}
\usepackage{booktabs}
\usepackage[capitalise]{cleveref}
\usepackage{color, colortbl}
\usepackage{booktabs}

\usepackage{nopageno}

\let\doendproof\endproof
\renewcommand\endproof{~\hfill$\qed$\doendproof}

\newenvironment{sketch}{\noindent{\itshape Sketch of proof.}}{~\hfill$\qed$\doendproof\smallskip}

\spnewtheorem{clm}{Claim}{\bfseries}{\rmfamily}

\newcommand{\myparagraph}[1]{\smallskip\noindent\textbf{\boldmath #1}}

\graphicspath{{}}

\definecolor{lightcyan}{rgb}{0.88,1,1}
\definecolor{antiquewhite}{rgb}{0.98, 0.92, 0.84}
\DeclareMathOperator\indeg{indeg}
\DeclareMathOperator\outdeg{outdeg}

\newcounter{casecounter}
\newcounter{subcasecounter}
\newcounter{subsubcasecounter}
\makeatletter
\newcommand{\ccase}[2][]{%
	\stepcounter{casecounter}%
	\setcounter{subcasecounter}{0}%
	\protected@write \@auxout {}{\string \newlabel {#2}{{#1\thecasecounter}{\thepage}{#1\thecasecounter}{#2}{}} }%
	\hypertarget{#2}{\noindent\textbf{Case #1\thecasecounter.}}
}

\newcommand{\subcase}[2][]{%
	\stepcounter{subcasecounter}%
	\setcounter{subsubcasecounter}{0}%
	\protected@write \@auxout {}{\string \newlabel {#2}{{#1\thecasecounter.\thesubcasecounter}{\thepage}{#1\thecasecounter.\thesubcasecounter}{#2}{}} }%
	\hypertarget{#2}{\noindent\textbf{Case #1\thecasecounter.\thesubcasecounter.}}
}

\newcommand{\subsubcase}[2][]{%
	\stepcounter{subsubcasecounter}%
	\protected@write \@auxout {}{\string \newlabel {#2}{{#1\thecasecounter.\thesubcasecounter.\thesubsubcasecounter}{\thepage}{#1\thecasecounter.\thesubcasecounter.\thesubsubcasecounter}{#2}{}} }%
	\hypertarget{#2}{\noindent\textbf{Case #1\thecasecounter.\thesubcasecounter.\thesubsubcasecounter.}}
}
\makeatother

\newcommand{\Pio}[2]{\textcolor{blue}{$\mathrm{I}_{#1}\mathrm{O}_{#2}$}}
\newcommand{\Pin}[1]{\textcolor{blue}{$\mathrm{I}_{#1}$}}

\begin{document}
	\title{Rectilinear Planarity Testing of Plane Series-Parallel Graphs in Linear Time \thanks{Work partially supported by: $(i)$ MIUR, grant 20174LF3T8 ``AHeAD: efficient Algorithms for HArnessing networked Data'', $(ii)$ Engineering Dep., Univ.~Perugia, grant RICBA19FM: ``Modelli, algoritmi e sistemi per la visualizzazione di grafi e reti''.}}
	\author{Walter Didimo\inst{1},
		Michael Kaufmann\inst{2},
		Giuseppe Liotta\inst{1},
		Giacomo Ortali\inst{1}$^{\textrm{\Letter}}$
	}

	\date{}
	
	\institute{
		Universit\`a degli Studi di Perugia, Italy\\
		\email {\{walter.didimo,giuseppe.liotta\}@unipg.it, giacomo.ortali@studenti.unipg.it}
		\and
		University of T\"ubingen, Germany\\
		\email {mk@informatik.uni-tuebingen.de}
	}

	\maketitle
	
	%
\begin{abstract}
A plane graph is \emph{rectilinear planar} if it admits an embedding-preserving straight-line drawing where each edge is either horizontal or vertical. We prove that rectilinear planarity testing can be solved in optimal $O(n)$ time for any plane series-parallel graph $G$ with $n$ vertices. If $G$ is rectilinear planar, an embedding-preserving rectilinear planar drawing of $G$ can be constructed in $O(n)$ time. 
Our result is based on a characterization of rectilinear planar series-parallel graphs in terms of intervals of orthogonal spirality that their components can have, and it leads to an algorithm that can be easily implemented.

\end{abstract}

\section{Introduction}\label{se:intro}
A \emph{planar orthogonal drawing} $\Gamma$ of a planar graph $G$ is a crossing-free drawing of $G$ that maps each vertex to a distinct point of the plane and each edge to a sequence of horizontal and vertical segments between its end-points~\cite{DBLP:books/ph/BattistaETT99,DBLP:reference/crc/DuncanG13,DBLP:books/ws/NishizekiR04}. A graph is \emph{rectilinear planar} if it admits a planar orthogonal drawing~without~bends.

Testing whether a graph is rectilinear planar is a fundamental question in graph drawing.  The problem can be either studied for \emph{plane} graphs, that is graphs that come with a fixed embedding, or in the variable embedding setting, where the algorithm  can choose one of the planar embeddings of the input graph. Besides being an interesting topic on its own right, rectilinear planarity testing is at the core of efficient algorithms that compute orthogonal drawings with minimum number of bends. For example, Rahman et al.~\cite{DBLP:journals/jgaa/RahmanNN03} characterize the rectilinear plane 3-graphs (i.e., graphs with vertex degree at most three) and then use this characterization to design linear time bend-minimization algorithms for these graphs in the fixed embedding setting~\cite{DBLP:journals/jgaa/RahmanNN99,DBLP:conf/wg/RahmanN02}. On the other hand, Garg and Tamassia~\cite{DBLP:conf/gd/GargL99} prove that rectilinear planarity testing is NP-complete for planar 4-graphs in the variable embedding setting.  Remarkably, the study of rectilinear plane 3-graphs has turned out to be an essential tool to design linear-time rectilinear planarity testing and bend-minimization algorithms for planar 3-graphs in the variable embedding setting~\cite{DBLP:conf/soda/DidimoLOP20,DBLP:conf/cocoon/Hasan019}.

\smallskip In this paper we study rectilinear planarity testing in the fixed embedding setting. A seminal paper of Tamassia~\cite{DBLP:journals/siamcomp/Tamassia87} implies that in this setting the problem can be solved in $O(n^2 \log n)$, where $n$ is the number of vertices of the input graph; its approach is based on solving a min-cost flow network problem to compute a bend-minimum orthogonal drawing of the input graph. Since its time of publication, establishing a lower bound on the time complexity of computing bend-minimum orthogonal drawings of plane graphs has remained a fascinating open problem (see, e.g, ~\cite{DBLP:conf/gd/BrandenburgEGKLM03,DBLP:books/ph/BattistaETT99,dlt-gd-17}). Garg and Tamassia~\cite{DBLP:conf/gd/GargT96a} improve the complexity to $O(n^{\frac{7}{4}} \sqrt{\log n})$ and then Cornelsen and Karrenbauer~\cite{DBLP:journals/jgaa/CornelsenK12} further improve the upper bound to $O(n^{1.5})$. For rectilinear planarity testing, the approach in~\cite{DBLP:journals/siamcomp/Tamassia87} reduces to compute a maximum flow in an $n$-vertex planar network with multiple sources and sinks; Borradaile et al.~\cite{DBLP:journals/siamcomp/BorradaileKMNW17} prove that this problem can be solved in $O(n \log^3 n)$ time.    
Since, as already mentioned, an $O(n)$-time algorithm for rectilinear planarity testing is known when the input is a plane 3-graph, the challenge is to understand whether an $O(n)$-time bound exists for plane 4-graphs. 

This paper sheds some light on this question by answering it for series-parallel graphs. An essential aspect of our approach is to tackle the problem without using any network-flow computation.
%
%
Our results are as follows:

\smallskip\noindent $(i)$ We give a characterization of those plane series-parallel graphs (with two terminals $s$ and $t$) that are rectilinear planar. This characterization is expressed in terms of values of \emph{spirality} that each series or parallel component can have in a rectilinear drawing. Intuitively, the spirality of a component measures how much it can be ``rolled-up'' in a rectilinear drawing of the graph.

\smallskip\noindent $(ii)$ While the possible values of spirality for each component may be linear, we can encode them in constant space. This makes it possible to design a linear-time rectilinear planarity testing algorithm for a two-terminal series-parallel graph $G$ based on a bottom-up visit of its decomposition tree $T$. If the test is positive, we compute in linear time a rectilinear drawing of $G$ through a top-down visit~of~$T$. The algorithm is easy to implement.

\smallskip The paper is organized as follows. Section~\ref{se:preli} recalls basic concepts. Section~\ref{se:characterization} gives our characterization of rectilinear planar series-parallel graphs in terms of their orthogonal spirality. Section~\ref{se:rect-alg-sp} describes the linear-time testing and drawing algorithm. Section~\ref{se:open} lists some open problems. In the main text of the paper some proofs are sketched or omitted, and can be found in the Appendix.

Together with our submission to GD 2020, another paper by Frati~\cite{frati-2020} was accepted to the same conference. The work of Frati is based on a different technique and it presents an $O(n)$-time algorithm for rectilinear planarity testing of outerplanar graphs. While the result of~\cite{frati-2020} does not apply to the family of graphs that are studied in this paper, it covers the variable embedding setting and the case of 1-connected outerplanar graphs.

\section{Preliminaries}\label{se:preli}

\myparagraph{Orthogonal Representations.} We focus on \emph{orthogonal representations} rather than orthogonal drawings. An orthogonal representation $H$ describes the shape of a class of orthogonal drawings in terms of sequences of bends along the edges and angles at the vertices. An (orthogonal) drawing $\Gamma$ of $H$ can be computed in linear time~\cite{DBLP:journals/siamcomp/Tamassia87}. If $H$ has no bend, it is a \emph{rectilinear representation} (see Fig.~\ref{fi:ortho}). The \emph{degree} $\deg (v)$ of a vertex $v$ denotes the number of edges incident to $v$.

\myparagraph{Series-Parallel graphs and Decomposition Trees.} A \textit{two-terminal series-parallel} graph, also called \emph{series-parallel graph} in the rest of the paper, has two distinct vertices $s$ and $t$, called its \textit{source} and its \textit{sink}, respectively, and it is inductively defined as follows: $(i)$~A single edge $(s,t)$ is a series-parallel graph with source~$s$~and~sink~$t$.
$(ii)$~Given $p \geq 2$ series-parallel graphs $G_1, \dots, G_p$, each $G_i$ with source $s_i$ and sink $t_i$ ($i=1, \dots, p$), a new series-parallel graph $G$ can be obtained with any of these two operations: {\em Series composition} -- It identifies $t_i$ with $s_{i+1}$ ($i=1,\dots,p-1$); $G$ has source $s=s_1$ and sink $t=t_p$. {\em Parallel composition} -- It identifies all sources $s_i$ together and all sinks $t_i$ together; $G$ has source~$s=s_i$~and~$t=t_i$~($i=1, \dots, p$).

\begin{figure}[tb]
	\centering
	\subfigure[$G$]{%
		\centering
		\includegraphics[height=0.42\columnwidth]{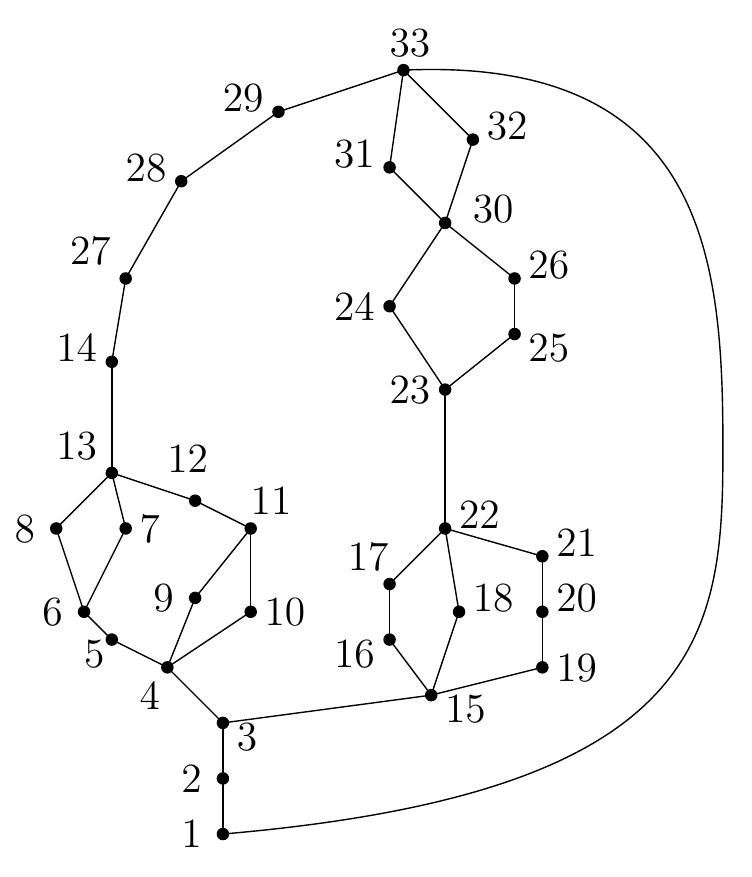}
		\label{fi:graph}
	}\hfil
	\subfigure[$H$]{%
		\centering
		\includegraphics[height=0.42\columnwidth]{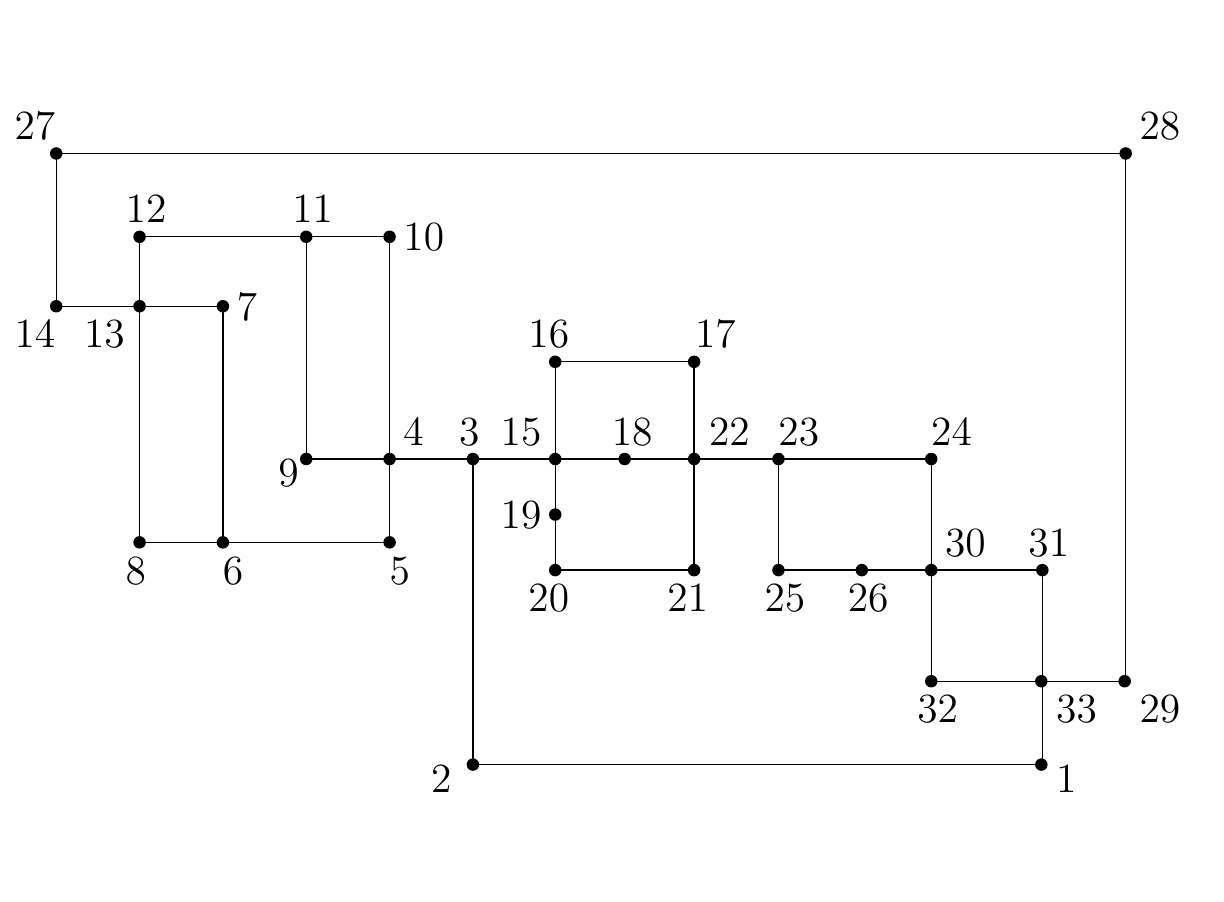}
		\label{fi:ortho}
	}\hfil
	\subfigure[$T$]{%
		\centering
		\includegraphics[width=0.97\columnwidth]{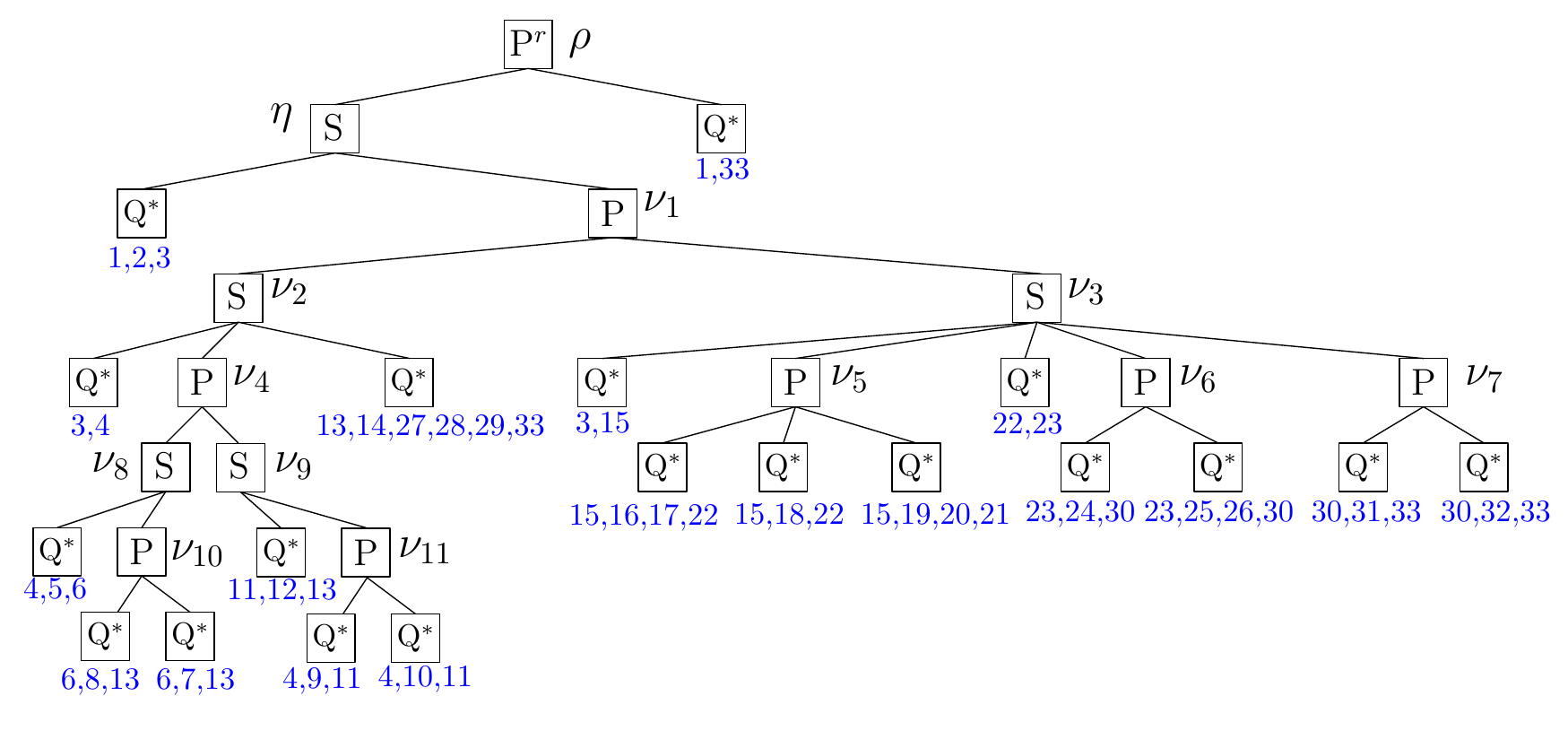}
		\label{fi:spq-tree}
	}
	\caption{(a) A biconnected series-parallel graph $G$. (b) A rectilinear planar representation $H$ of $G$. (c) The SPQ$^*$-tree $T$ of $G$ with reference edge $(1,33)$.}\label{fi:prel}
\end{figure}
	
A series-parallel graph $G$ is naturally associated with a \emph{decomposition tree}~$T$, which describes the series and parallel compositions that build $G$. Tree $T$ has three types of nodes:
S-, P-, and Q$^*$-nodes. If $G$ is the series composition of $p \geq 2$ graphs $G_i$ that are not all single edges, the root of $T$ is an S-node whose subtrees are the decomposition trees $T_i$ of $G_i$. If $G$ is the parallel composition of $p \geq 2$ graphs $G_i$, the root of $T$ is a P-node whose subtrees are the decomposition trees $T_i$ of $G_i$. If $G$ is a series composition of $\ell \geq 1$ edges, its decomposition tree is a single Q$^*$-node and for brevity we say that $\ell$ is the \emph{length} of this node.

For a node $\nu$ of $T$, the \emph{pertinent graph} $G_\nu$ of $\nu$ is the series-parallel subgraph of $G$ formed by all edges associated with the Q$^*$-nodes in the subtree rooted at~$\nu$. We also call $G_\nu$ a \emph{component} of $G$. If $u$ and $v$ are the source and the sink of $G_\nu$, respectively, we say that $\{u,v\}$ are the \emph{poles} of $G_\nu$ and of $\nu$: $u$ is the \emph{source pole} and $v$ is the \emph{sink pole}. If $G$ is a biconnected plane series-parallel graph, for any edge $e=(s,t)$ on the external face of $G$, we can associate with $G$ a decomposition tree $T$ where the root is a P-node representing the parallel composition between $e$ and the rest of the graph. Thus, the root of $T$ is always a P-node with two children, one of which is a Q$^*$-node corresponding to $e$. It will be called the (unique) P$^r$-node of $T$, to distinguish it by the other P-nodes. Edge $e$ is the \emph{reference edge} of $T$ and $T$ is the SPQ$^*$-tree of $G$ \emph{with respect to} $e$.
Also, it is always possible to make $T$ such that each P-node (distinct from the root) has no P-node child and each S-node has no S-node child. Since we only deal with graphs of vertex-degree at most four, a P-node has either two or three children. From now on we assume that $T$ always satisfies the properties above for a biconnected series-parallel graph. Observe that the number of nodes of $T$ is $O(n)$, where $n$ is the number of vertices of $G$. Figure~\ref{fi:prel} shows a biconnected series-parallel graph $G$, a rectilinear planar representation~$H$~of~$G$, and the SPQ$^*$-tree~$T$~of~$G$~with~respect~to~the~reference~edge~$(1,33)$.





\section{Characterizing Rectilinear Plane Series-Parallel Graphs}\label{se:characterization}

Let $G$ be a plane series-parallel graph. If $G$ is biconnected let $e=(s,t)$ be any edge on the external face of~$G$; otherwise, by definition of two-terminal series-parallel graph, we can add a dummy edge $e$ on the external face of $G$ to make it biconnected. We assume that the external face of $G$ is to the right of $e$ while moving from $s$ to $t$ (as in Fig.~\ref{fi:graph}). Let $T$ be an SPQ$^*$-tree of $G$ with respect to~$e$. An overview of our algorithm is as follows. It visits $T$ in post-order (a node is visited after its children). When the algorithm visits a node $\nu$, it tests whether $G_\nu$ admits a planar rectilinear representation by checking whether a certain condition, which we call \emph{representability condition}, is verified: In the negative case, the algorithm halts and rejects the instance; else it stores in $\nu$ its \emph{representability interval} $I_\nu$. Such an interval is a compact representation of the possible values of \emph{orthogonal spirality} that the pertinent graph $G_{\nu}$ of $\nu$ may have in a rectilinear representation of $G$. Informally speaking, the orthogonal spirality is a measure of how much a rectilinear representation of pertinent graph $G_{\nu}$ is ``rolled-up'' in a rectilinear planar representation of $G$. 
As we shall see, the representability interval is such that for every value $k \in I_\nu$ graph $G_\nu$ admits a planar rectilinear representation with spirality $k$, while it does not for any value outside $I_\nu$. If the testing algorithm does not halt and it reaches the root, two cases are considered: If $e$ is a real edge of $G$, then the algorithm executes a final test to check whether a rectilinear planar representation of $G$ can be obtained by merging a straight-line representation of $e$ with a rectilinear representation of the child component of the root other than $e$. If $e$ is a dummy edge added to make $G$ biconnected this check is not required, because~$e$~is~not~present~in~the~final~representation~and~can~arbitrarily~bend.

We now present the characterization of the rectilinear planar components in terms of representability conditions and intervals that is at the base of the~testing algorithm. We start in Section~\ref{sse:spirality} with a formal definition of spirality. We characterize Q$^*$-, S-, and P-components with three children in Section~\ref{sse:Q-S-P3-representability}, and P-components with two children in Section~\ref{sse:P2-representability}. We summarize in Section~\ref{sse:characterization}.


\subsection{Spirality of Series-Parallel Graphs}\label{sse:spirality}

Let $T$ be an SPQ$^*$-tree of a biconnected plane series-parallel graph $G$ for a given reference edge $e=(s,t)$. Let $H$ be an embedding-preserving orthogonal representation of $G$. Also, let $\nu$ be a node of $T$ with poles $\{u,v\}$, and let $H_\nu$ be the restriction of $H$ to the pertinent graph $G_\nu$ of $\nu$. We also say that $H_\nu$ is a \emph{component} of $H$. For each pole $w \in \{u,v\}$, let  $\indeg_\nu(w)$ and $\outdeg_\nu(w)$ be the degree of $w$ inside and outside $H_\nu$, respectively. Define two (possibly coincident) \emph{alias vertices} of $w$, denoted by $w'$ and $w''$, as follows:
$(i)$ if $\indeg_\nu(w)=1$, then $w'=w''=w$;
$(ii)$ if $\indeg_\nu(w)=\outdeg_\nu(w)=2$, then $w'$ and $w''$ are dummy vertices, each splitting one of the two distinct edge segments incident to $w$ outside $H_\nu$;
$(iii)$ if $\indeg_\nu(w)>1$ and $\outdeg_\nu(w)=1$, then $w'=w''$ is a dummy vertex that splits the edge segment incident to $w$ outside $H_\nu$.

Let $A^w$ be the set of distinct alias vertices of a pole $w$. Let $P^{uv}$ be any simple path from $u$ to $v$ inside $H_\nu$ and let $u' \in A^u$ and $v' \in A^v$. The path $S^{u'v'}$ obtained concatenating $(u',u)$, $P^{uv}$, and $(v,v')$ is called a \emph{spine} of $H_\nu$. Denote by $n(S^{u'v'})$ the number of right turns minus the number of left turns encountered along $S^{u'v'}$ while moving from $u'$ to $v'$.
%
The \emph{spirality} $\sigma(H_\nu)$ of $H_\nu$ is defined based on the following cases:
	$(a)$ $A^u=\{u'\}$ and $A^v=\{v'\}$. Then $\sigma(H_\nu) = n(S^{u'v'})$.
	$(b)$ $A^u=\{u'\}$ and $A^v=\{v',v''\}$. Then $\sigma(H_\nu) = \frac{n(S^{u'v'}) + n(S^{u'v''})}{2}$.
	$(c)$ $A^u=\{u',u''\}$ and $A^v=\{v'\}$. Then $\sigma(H_\nu) = \frac{n(S^{u'v'}) + n(S^{u''v'})}{2}$.
	$(d)$ $A^u=\{u',u''\}$ and $A^v=\{v',v''\}$. Without loss of generality, assume that $(u,u')$ precedes $(u,u'')$ counterclockwise around $u$ and that $(v,v')$ precedes $(v,v'')$ clockwise around $v$. Then $\sigma(H_\nu) = \frac{n(S^{u'v'}) + n(S^{u''v''})}{2}$.
%
Notice that, by definition, the spirality of $H_\nu$ also depends on the angles at the poles of $H_\nu$, not only on the shape of $H_\nu$. 

\smallskip Di Battista et al.~\cite{DBLP:journals/siamcomp/BattistaLV98} showed that the spirality of $H_\nu$ does not vary with the choice of path $P^{uv}$ and that two distinct representations of $G_\nu$ with the same spirality are interchangeable. Fig.~\ref{fi:spiralities} reports the spiralities of some P- and S-components in the representation $H$ of Fig.~\ref{fi:ortho}. For brevity, we shall denote by $\sigma_\nu$ the spirality of an orthogonal representation~of~$G_\nu$. 

\begin{figure}[!h]
	\centering
	\includegraphics[width=1\columnwidth]{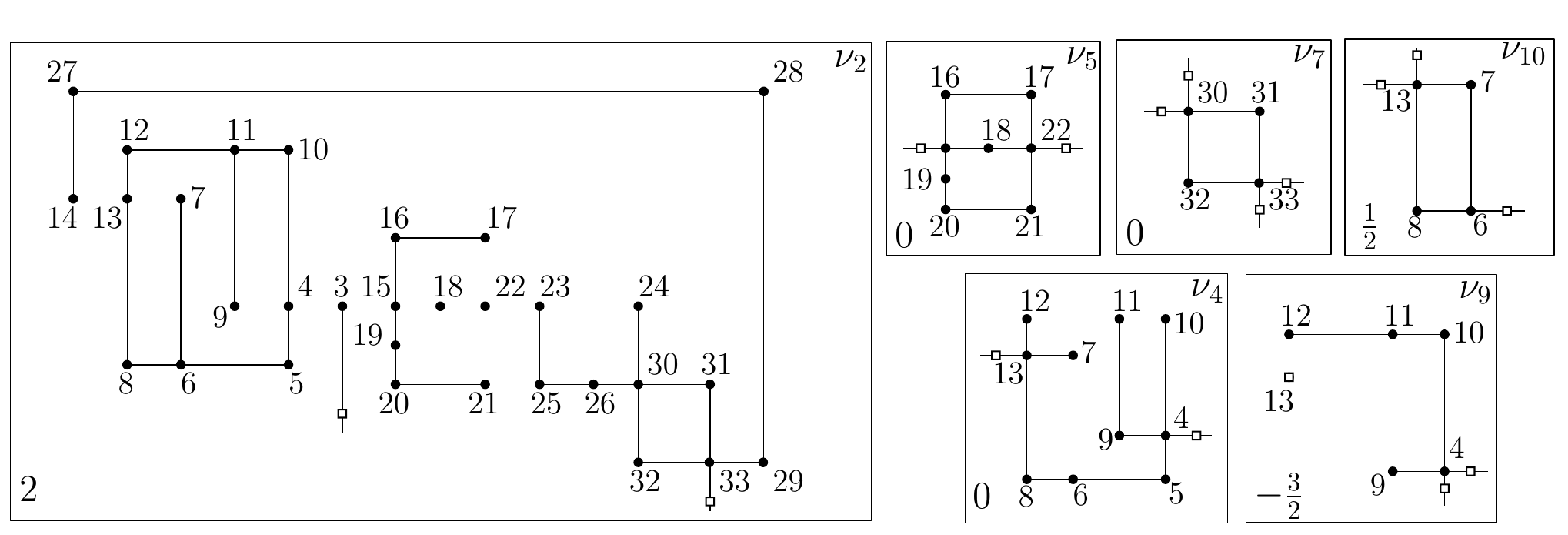}
	\caption{Spiralities (left-bottom corners) of some components in the representation $H$ of Fig.~\ref{fi:ortho}. Small squares indicate alias vertices.}\label{fi:spiralities}
\end{figure}

%

\begin{lemma}[\cite{DBLP:journals/siamcomp/BattistaLV98}]\label{le:spirality-S-node}
	Let $\nu$ be an S-node of $T$ with children $\mu_1, \dots, \mu_h$. The following relationship holds: $\sigma_\nu = \sum_{i=1}^{h}\sigma_{\mu_i}.$
\end{lemma}	

If $\nu$ is a P-node with two children, we denote by $\mu_l$ and $\mu_r$ the left child and the right child of $\nu$, respectively. If $\nu$ is a P-node with three children, we denote by $\mu_l$, $\mu_c$, and $\mu_r$, the three children of $\nu$ from left to right. Also, for each pole $w \in \{u,v\}$ of $\nu$, the \emph{leftmost angle} at $w$ in $H$ is the angle formed by the leftmost external edge and the leftmost internal edge of $H_\nu$ incident to $w$. The \emph{rightmost angle} at $w$ in $H$ is defined symmetrically.
We define two binary variables $\alpha_w^l$ and $\alpha_w^r$ as follows: $\alpha_w^l = 0$ ($\alpha_w^r = 0$) if the leftmost (rightmost) angle at $w$ in $H$ is of $180^\circ$, while $\alpha_w^l = 1$ ($\alpha_w^r = 1$) if this angle is of $90^\circ$.
Observe that if $\deg(w)=4$, then $\alpha_w^l = \alpha_w^r = 1$. Also, if $\nu$ has two children, define two additional variables $k_w^l$ and $k_w^r$ as follows: $k_w^d = 1$ if $\indeg_{\mu_d}(w)=\outdeg_{\nu}(w)=1$,
while $k_{w}^d=1/2$ otherwise, for $d \in \{l,r\}$. 

For example, in Fig.~\ref{fi:spiralities} the P-component of $\nu_4$  has poles $u=4$ and $v=13$, and we have $k_u^l =k_v^r=1$, $k_u^r =k_v^l=\frac{1}{2}$, and $\alpha_u^l=\alpha_u^r=\alpha_v^l=\alpha_v^r=1$. The P-component of $\nu_{10}$ has poles $u=6$ and $v=13$, and we have $k_u^l=k_u^r=1$, $k_v^l=k_v^r=\frac{1}{2}$, $\alpha_u^l=0$, and $\alpha_u^r=\alpha_v^l=\alpha_v^r=1$. Fig.~\ref{fi:P-node-types} reports all the values of $k_{w}^d$ for the possible types of P-nodes with two children.

\begin{lemma}[\cite{DBLP:journals/siamcomp/BattistaLV98}]\label{le:spirality-P-node-2-children}
	Let $\nu$ be a P-node of $T$ with two children $\mu_l$ and $\mu_r$. The following relationships hold:
	$\sigma_\nu = \sigma_{\mu_l} - k_{u}^l \alpha_{u}^l - k_{v}^l \alpha_{v}^l = \sigma_{\mu_r} + k_{u}^r \alpha_{u}^r + k_{v}^r\alpha_{v}^r$.
\end{lemma}


\begin{lemma}[\cite{DBLP:journals/siamcomp/BattistaLV98}]\label{le:spirality-P-node-3-children}
	Let $\nu$ be a P-node of $T$ with three children $\mu_l$, $\mu_c$, and $\mu_r$. The following relationships hold:
	$\sigma_\nu = \sigma_{\mu_l} - 2 = \sigma_{\mu_c} = \sigma_{\mu_r} + 2$.
\end{lemma}

About the values of spirality $\sigma_\nu$ that a component $H_\nu$ can take, if $\nu$ is a Q$^*$-node or a P-node with three children, $\sigma_\nu$ is always an integer. If $\nu$ is an S-node or a P-node with two children, $\sigma_\nu$ is either integer or semi-integer depending on whether the total number of alias vertices for the poles of $\nu$ is even or odd.



\subsection{Q$^*$-nodes, S-nodes, and P-nodes with three children}\label{sse:Q-S-P3-representability}

From now on, when we say that the spirality $\sigma_\nu$ of an orthogonal planar representation of $G_\nu$ can take \emph{all} values in an interval $[a,b]$, we mean that such values are either all the integer numbers or all the semi-integer numbers in $[a,b]$, depending on the cases described above for $\nu$.

\begin{lemma}\label{le:Q*-representability}
	Let $\nu$ be a Q$^*$-node of length $\ell$. Graph $G_\nu$ is always rectilinear planar (i.e., its representability condition is always true) and its representability interval is $I_\nu = [-\ell+1, \ell-1]$.
\end{lemma}
\begin{proof}
	$G_\nu$ is a path with $\ell-1$ degree-2 vertices. For any integer $k \in [-\ell+1,0]$, a rectilinear planar representation $H_\nu$ of $G_\nu$ with spirality $k$ is obtained by making a left turn at $k$ degree-2 vertices of $G_\nu$ (going from the source to the sink pole), and no turn at any remaining vertex of $G_\nu$. Symmetrically, for any $k \in (0,\ell-1]$, we realize $H_\nu$ with spirality $k$ by making a right turn at exactly $k$ degree-2 vertices of $G_\nu$. It is clear that no values of spirality out of $I_\nu$ can be achieved.
\end{proof}	

\begin{lemma}\label{le:S-representability}
	Let $\nu$ be an S-node with $h$ children $\mu_1, \dots, \mu_h$. Suppose that, for every $i \in [1,h]$, the representability interval of $G_{\mu_i}$ is $I_{\mu_i} = [m_i,M_i]$. Graph $G_\nu$ is always rectilinear planar (i.e., its representability condition is always true) and its representability interval is $I_\nu = [\sum_{i=1}^hm_i,\sum_{i=1}^hM_i]$.
\end{lemma}	
\begin{proof}
	We use induction on the number of children of $\nu$. In the base case $h=2$. By hypothesis $I_{\mu_1}=[m_1, M_1]$ and $I_{\mu_2}=[m_2,M_2]$. By Lemma~\ref{le:spirality-S-node}, a series composition of a rectilinear representation of $G_{\mu_1}$ with spirality $\sigma_{\mu_1}$ and of a rectilinear representation of $G_{\mu_2}$ with spirality $\sigma_{\mu_2}$ results in a rectilinear representation of $G_\nu$ with spirality $\sigma_\nu = \sigma_{\mu_1} + \sigma_{\mu_2}$. Hence, if $M_1 = m_1+r_1$ and $M_2 = m_2+r_2$, for two non-negative integers $r_1$ and $r_2$, then the possible values for $\sigma_\nu$ are exactly $m_1+m_2, m_1+1+m_2, \dots, m_1+r_1+m_2, \dots, m_1+r_1+m_2+1, \dots, m_1+r_1+m_2+r_2$, i.e., all values in the interval $[m_1+m_2, M_1+M_2]$. In the inductive case $h \geq 3$; consider the series composition $G'_1$ of $G_{\mu_1}, \dots, G_{\mu_{h-1}}$. Graph $G_\nu$ is the series composition of $G'_1$ and $G_{\mu_2}$. By inductive hypothesis the representability interval of $G'_1$ is $[\sum_{i=1}^{h-1}m_i,\sum_{i=1}^{h-1}M_i]$ and by Lemma~\ref{le:spirality-S-node} applied to $G'_1$ and $G_{\mu_2}$ we have $I_\nu = [\sum_{i=1}^hm_i,\sum_{i=1}^hM_i]$, using the same reasoning as for the base case.
\end{proof}

\begin{lemma}\label{le:P-3-children-representability}
	Let $\nu$ be a P-node with three children $\mu_l$, $\mu_c$, and $\mu_r$. Suppose that $G_{\mu_l}, G_{\mu_c}$, and $G_{\mu_r}$ are rectilinear planar and that their representability intervals are $I_{\mu_l}=[m_l, M_l]$, $I_{\mu_c}=[m_c, M_c]$, and $I_{\mu_r} = [m_r, M_r]$, respectively. Graph $G_\nu$ is rectilinear planar if and only if $[m_l-2,M_l-2] \cap [m_c,M_c] \cap [m_r+2,M_r+2] \neq \emptyset$. Also, if this representability condition holds then the representability interval of $G_\nu$ is $I_\nu = [\max\{m_l-2,m_c,m_r+2\},\min\{M_l-2,M_c,M_r+2\}]$.
\end{lemma}
\begin{proof}
	\noindent\textsf{Representability condition.} Suppose first that $G_\nu$ is rectilinear planar and let $H_\nu$ be a rectilinear planar representation of $G_\nu$ with spirality $\sigma_\nu$. By Lemma~\ref{le:spirality-P-node-3-children}, the spiralities $\sigma_{\mu_l}$, $\sigma_{\mu_c}$, and $\sigma_{\mu_r}$ for the representations of $G_{\mu_l}, G_{\mu_c}$, and $G_{\mu_r}$ in $H_\nu$ are such that $\sigma_{\mu_l}=\sigma_\nu+2$, $\sigma_{\mu_c}=\sigma_\nu$, and $\sigma_{\mu_r}=\sigma_\nu-2$. Since $\sigma_{\mu_l} \in [m_l,M_l]$, $\sigma_{\mu_c} \in [m_c,M_c]$, $\sigma_{\mu_r} \in [m_r,M_r]$, we have $\sigma_\nu \in [m_l-2,M_l-2] \cap [m_c,M_c] \cap [m_r+2,M_r+2]$.
	Suppose vice versa that $[m_l-2,M_l-2] \cap [m_c,M_c] \cap [m_r+2,M_r+2] \neq \emptyset$, and let $k$ be any value in such intersection. Setting $\sigma_{\mu_l} = k + 2$, $\sigma_{\mu_c} = k$, and $\sigma_{\mu_r} = k - 2$ we have $\sigma_{\mu_l} \in [m_l,M_l]$, $\sigma_{\mu_c} \in [m_c,M_c]$, and $\sigma_{\mu_r} \in [m_r,M_r]$. By Lemma~\ref{le:spirality-P-node-3-children}, $G_\nu$ is rectilinear planar for a value of spirality~$\sigma_\nu = k$.
	
	\smallskip\noindent\textsf{Representability interval.} Assume that $G_\nu$ is rectilinear planar. Clearly $[\max\{m_l-2,m_c,m_r+2\},\min\{M_l-2,M_c,M_r+2\}] = [m_l-2,M_l-2] \cap [m_c,M_c] \cap [m_r+2,M_r+2]$, and by the truth of the feasiblity condition we have $[\max\{m_l-2,m_c,m_r+2\},\min\{M_l-2,M_c,M_r+2\}] \neq \emptyset$.
	Similarly to the first part of the proof of the representability condition, any rectilinear planar representation of $G_\nu$ has a value of spirality in the interaval $[\max\{m_l-2,m_c,m_r+2\},\min\{M_l-2,M_c,M_r+2\}]$.
	On the other hand, let $k \in [\max\{m_l-2,m_c,m_r+2\},\min\{M_l-2,M_c,M_r+2\}]$. Analogously to the second part of the proof of the representability condition, we can construct a rectilinear planar representation of $G_\nu$ with spirality $\sigma_\nu=k$, by combining in parallel rectilinear planar representations of $G_{\mu_l}$, $G_{\mu_c}$, and $G_{\mu_r}$ with spiralities $\sigma_{\mu_l} = \sigma_\nu + 2$, $\sigma_{\mu_c} = \sigma_\nu$, and $\sigma_{\mu_r} = \sigma_\nu - 2$, respectively.
\end{proof}

\subsection{P-nodes with two children}\label{sse:P2-representability}
For a P-node $\nu$ with two children $\mu_l$ and $\mu_r$, the representability condition and interval depend on the indegree and outdegree of the poles of $\nu$ in $G_\nu$, $G_{\mu_l}$, and $G_{\mu_r}$. We define the \emph{type} of $\nu$ and of $G_\nu$ as follows (refer to Fig.~\ref{fi:P-node-types}):

\smallskip\noindent -- \Pio{2}{\alpha\beta}: Both poles of $\nu$ have indegree two in $G_\nu$; also one pole has outdegree $\alpha$ in $G_\nu$ and the other pole has outdegree $\beta$ in $G_\nu$, for $1 \leq \alpha \leq \beta \leq 2$. This gives rise to the specific types \Pio{2}{11}, \Pio{2}{12}, and \Pio{2}{22}.

\smallskip\noindent -- \Pio{3d}{\alpha\beta}: One pole of $\nu$ has indegree two in $G_\nu$, while the other pole has indegree three in $G_\nu$ and indegree two in $G_{\mu_d}$ for $d \in \{l,r\}$; also one pole has outdegree $\alpha$ in $G_\nu$ and the other has outdegree $\beta$ in $G_\nu$, for $1 \leq \alpha \leq \beta \leq 2$, where $\alpha=\beta=2$ is not possible. This gives rise to the specific types \Pio{3l}{11}, \Pio{3r}{11}, \Pio{3l}{12}, \Pio{3r}{12}.
	
\smallskip\noindent -- \Pin{3dd'}: Both poles of $\nu$ have indegree three in $G_\nu$; one of the two poles has indegree two in $G_{\mu_d}$ and the other has indegree two in $G_{\mu_{d'}}$, for  $dd' \in \{ll,lr,rr\}$ (both poles have outdegree one in $G_\nu$). Hence, the specific types are~\Pin{3ll},~\Pin{3lr},~\Pin{3rr}.

\begin{figure}[!h]
	\centering
	\includegraphics[width=1\columnwidth]{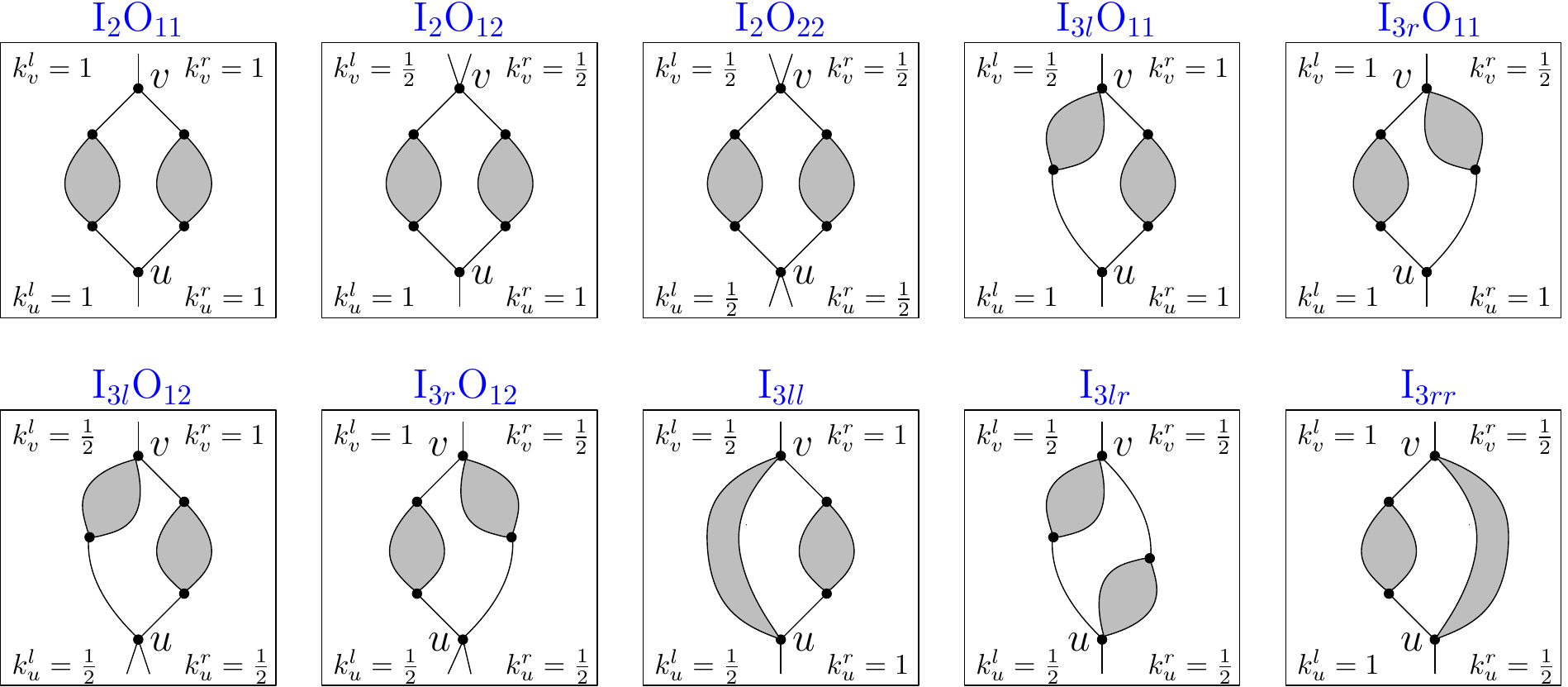}
	\caption{Schematic illustration of the different types of P-nodes with two children.}\label{fi:P-node-types}
\end{figure}



\noindent To characterize P-nodes of type \Pio{2}{\alpha\beta} we start with the following result.


\begin{restatable}{lemma}{itwotypeonesupport}\label{le:P-2-children-support-type1}
	Let $G_\nu$ be a P-node of type \Pio{2}{\alpha\beta} with children $\mu_l$ and $\mu_r$. $G_\nu$ is rectilinear planar if and only if $G_{\mu_l}$ and $G_{\mu_r}$ are rectilinear planar for values of spiralities $\sigma_{\mu_l}$ and $\sigma_{\mu_r}$ such that $\sigma_{\mu_l}-\sigma_{\mu_r}\in[2,4-\gamma]$, where $\gamma = \alpha + \beta - 2$.
\end{restatable}
\begin{sketch}
	We only give the proof for $\alpha=\beta=2$. The other cases are treated similarly (see the appendix). In this case $G_\nu$ is of type \Pio{2}{22} and we prove that $G_\nu$ is rectilinear planar if and only if $G_{\mu_l}$ and $G_{\mu_r}$ are rectilinear planar for values of spiralities $\sigma_{\mu_l}$ and $\sigma_{\mu_r}$ such that $\sigma_{\mu_l}-\sigma_{\mu_r}=2$. We have $k_u^l=k_u^r=\frac{1}{2}$.
	If $G_\nu$ is rectilinear planar, we have that $\alpha_{u}^l+\alpha_{u}^r = \alpha_{v}^l+\alpha_{v}^r = 2$. By Lemma~\ref{le:spirality-P-node-2-children}, $\sigma_{\mu_l} = \sigma_{\nu} + 1$ and $\sigma_{\mu_r} = \sigma_{\nu} - 1$; hence $\sigma_{\mu_l}-\sigma_{\mu_r}=2$.
	
	Suppose vice versa that $\sigma_{\mu_l}-\sigma_{\mu_r}=2$. We show that $G_\nu$ admits a rectilinear planar representation $H_\nu$. We obtain $H_\nu$ by combining in parallel the two rectilinear planar representations of $G_{\mu_l}$ and $G_{\mu_r}$ and by suitably setting $\alpha_u^d$ and $\alpha_v^d$ ($d \in \{l,r\}$). For any cycle $C$ through $u$ and $v$, the number of $90^\circ$ angles minus the number of $270^\circ$ angles in the interior of $C$ can be expressed by $a_c = \sigma_{\mu_l}-\sigma_{\mu_r} + 1 + 1$ (both the angles at $u$ and $v$ inside $C$ is always of 90$^\circ$ degrees). We then set $\alpha_u^l=\alpha_v^l=\alpha_u^r=\alpha_v^r=1$, which guarantees $a_c = 4$.
	Also, any other cycle not passing through $u$ and $v$ is an orthogonal polygon because it belongs to a rectilinear planar representation of either $G_{\mu_l}$ or $G_{\mu_r}$.
\end{sketch}

\begin{restatable}{lemma}{itwotypeone}\label{le:P-2-children-representability-type1}
	Let $\nu$ be a P-node of type \Pio{2}{\alpha\beta} with children $\mu_l$ and $\mu_r$. Suppose that $G_{\mu_l}$ and $G_{\mu_r}$ are rectilinear planar with representability intervals $I_{\mu_l}=[m_l, M_l]$ and $I_{\mu_r} = [m_r, M_r]$, respectively. Graph $G_\nu$ is rectilinear planar if and only if $[m_l-M_r,M_l-m_r] \cap [2,4-\gamma] \neq \emptyset$, where $\gamma = \alpha + \beta -2$. Also, if this representability condition holds then the representability interval of $G_\nu$ is $I_\nu = [\max\{m_l-2,m_r\}+\frac{\gamma}{2}, \min\{M_l, M_r+2\}-\frac{\gamma}{2}]$.
\end{restatable}
\begin{sketch}
	We consider the case $\alpha=\beta=2$. The other cases are treated similarly (see the appendix).
	In this case $G_\nu$ is of type \Pio{2}{22} and we prove that $I_\nu = [\max\{m_l-2,m_r\}+1, \min\{M_l, M_r+2\}-1]$.
	
	Assume first that $G_\nu$ is rectilinear planar and let $H_\nu$ be a rectilinear planar representation of $G_\nu$ with spirality $\sigma_\nu$. Let $H_{\mu_l}$ and $H_{\mu_r}$ be the rectilinear planar representations of $G_{\mu_l}$ and $G_{\mu_r}$ contained in $H_\nu$, and let  $\sigma_{\mu_l}$ and $\sigma_{\mu_r}$ their spiralities. Since both $u$ and $v$ have outdegree two in $G_\nu$ we have that $\alpha_{u}^l+\alpha_{u}^r = \alpha_{v}^l+\alpha_{v}^r = 2$. By Lemma~\ref{le:spirality-P-node-2-children},  $\sigma_{\mu_l} = \sigma_{\nu} + 1$ and $\sigma_{\mu_r} = \sigma_{\nu} - 1$. By the representability condition $\sigma_{\mu_r}=\sigma_{\mu_l}-2$. Hence $\sigma_{\mu_r}\ge m_l-2$ and $\sigma_{\mu_r}\ge \max\{m_l-2,m_r\}$.
	Also by $\sigma_{\nu}=\sigma_{\mu_r}+1$, $\sigma_\nu\ge \max\{m_l-2,m_r\}+1$. Similarly, by the representability condition $\sigma_{\mu_l}=\sigma_{\mu_r}+2$. Hence $\sigma_{\mu_l}\le M_r+2$ and $\sigma_{\mu_l}\le \max\{M_l,M_r+2\}$. Since $\sigma_{\mu_l}=\sigma_\nu+1$ we have $\sigma_\nu\le \max\{M_l,M_r+2\}-1$.
	
	Assume vice versa that $k$ is an integer in the interval $I_\nu = [\max\{m_l-2,m_r\}+1,\min\{M_l,M_r+2\}-1]$. We show that there exists a rectilinear planar representation of $G_\nu$ with spirality $\sigma_\nu=k$. We have  $k +1 \in [\max\{m_l,m_r+2\},\min\{M_l,M_r+2\}]$ and therefore $k +1 \in [m_l,M_l]$. Hence there is a rectilinear planar representation $H_{\mu_l}$ of $G_{\mu_l}$ with spirality $\sigma_{\mu_l}=k+1$. Similarly, $k -1 \in [\max\{m_l-2,m_r\},\min\{M_l-2,M_r\}]$ and therefore $k-1 \in [m_r,M_r]$. Hence there is a rectilinear planar representation $H_{\mu_r}$ of $G_{\mu_r}$ with spirality $\sigma_{\mu_r} = k-1$. By the representability condition, $G_\nu$ has a rectilinear planar representation $H_\nu$; with the same construction as in Lemma~\ref{le:P-2-children-support-type1},~the~spirality~of~$H_\nu$~is~$\sigma_\nu=k$.
\end{sketch}

\medskip The proofs of the next lemmas are similar to Lemma~\ref{le:P-2-children-representability-type1} (see the appendix).

\begin{restatable}{lemma}{itwotypetwo}\label{le:P-2-children-representability-type2}
	Let $\nu$ be a P-node of type \Pio{3d}{\alpha\beta} with children $\mu_l$ and $\mu_r$.
	Suppose that $G_{\mu_l}$ and $G_{\mu_r}$ are rectilinear planar with representability intervals $I_{\mu_l}=[m_l, M_l]$ and $I_{\mu_r} = [m_r, M_r]$, respectively. Graph $G_\nu$ is rectilinear planar if and only if $[m_l-M_r,M_l-m_r] \cap [\frac{5}{2},\frac{7}{2}-\gamma] \neq \emptyset$, where $\gamma = \alpha + \beta -2$. Also, if this representability condition holds then the representability interval of $G_\nu$ is $I_\nu = [\max\{m_l-\frac{3}{2},m_r+1\}+\frac{\gamma-\rho(d)}{2},\min\{M_l-\frac{1}{2}, M_r+2\}-\frac{\gamma+\rho(d)}{2}]$, where $\rho(\cdot)$ is a function such that $\rho(r)=1$ and $\rho(l)=0$.
\end{restatable}

\begin{restatable}{lemma}{itwotypethree}\label{le:P-2-children-representability-type3}
	Let $\nu$ be a P-node of type \Pin{3dd'} with children $\mu_l$ and $\mu_r$.
	Suppose that $G_{\mu_l}$ and $G_{\mu_r}$ are rectilinear planar with representability intervals $I_{\mu_l}=[m_l, M_l]$ and $I_{\mu_r} = [m_r, M_r]$, respectively. Graph $G_\nu$ is rectilinear planar if and only if $3 \in [m_l-M_r,M_l-m_r]$. Also, if this representability condition holds then the representability interval of $G_\nu$ is $I_\nu = [\max\{m_l-1,m_r+2\}-\frac{\rho(d)+\rho(d')}{2},\min\{M_l-1, M_r+2\}-\frac{\rho(d)+\rho(d')}{2}]$, where $\rho(\cdot)$ is a function such that $\rho(r)=1$ and $\rho(l)=0$.
\end{restatable}


%

\subsection{Characterization}\label{sse:characterization}

Lemmas~\ref{le:Q*-representability},~\ref{le:S-representability},~\ref{le:P-3-children-representability},~\ref{le:P-2-children-representability-type1},~\ref{le:P-2-children-representability-type2}, and~\ref{le:P-2-children-representability-type3} give rise to the following characterization.  


\begin{theorem}\label{th:representability-conditions-intervals}
	Let $G$ be a plane series-parallel graph and let $T$ be an SP$Q^*$-tree of $G$. Let $\nu$ be any non-root node of $T$. The plane graph $G_\nu$ is rectilinear planar if and only if it satisfies the representability condition given in Table~\ref{ta:representability}. Also, if such condition is satisfied, $G_\nu$ admits a rectilinear planar representation for all and only the values of spirality in the representability interval given in Table~\ref{ta:representability}.
\end{theorem}

To finally achieve a characterization of rectilinear series-parallel graphs we need to consider the representability condition that must be verified at the level of the root, when the reference edge is not a dummy edge. Denote by $e=(u,v)$ be the reference edge of $G$ and let $\rho$ be the root of $T$ with respect to $e$. Let $\eta$ be the child of $\rho$ that does not correspond to $e$, and let $u'$ and $v'$ be the alias vertices associated with the poles $u$ and $v$ of $G_\eta$. Suppose that $G_\eta$ is rectilinear planar with representability interval $I_\eta$. 

We say that $G$ satisfies the \emph{root condition} if $I_\eta \cap \Delta_\rho \neq \emptyset$, where $\Delta_\rho$ is defined as follows: $(i)$ $\Delta_\rho=[2,6]$ if $u'$ coincides with $u$ and $v'$ coincides with $v$; $(ii)$ $\Delta_\rho=[3,5]$ if exactly one of $u'$ and $v'$ coincides with $u$ and $v$, respectively; $(iii)$ $\Delta_\rho=4$ if none of $u'$ and $v'$ coincides with $u$ and $v$.

\begin{lemma}\label{le:root-condition}
	Let $e=(u,v)$ be the reference edge of $G$ and let $\rho$ be the root of $T$ with respect to $e$.
	Let $\eta$ be the child of $\rho$ that does not correspond to $e$. Suppose that $G_\eta$ is rectilinear planar with representability interval $I_\eta$. $G$ is rectilinear planar if and only if it satisfies the root condition. Also, if $G$ satisfies the root condition, it admits a rectilinear planar representation $H$ for any value of spirality $\sigma_\eta$ of $H_\eta$ such that $\sigma_\eta \in I_\eta \cap \Delta_\rho$, where $H_\eta$ is the restriction of $H$ to $G_\eta$.     
\end{lemma}
\begin{proof}
	Let $f_{int}$ be the internal face of $G$ incident to $e$. Observe that $u$ and $v$ are the poles of $G_\eta$. Let $u'$ be the alias vertex associated with $u$ and let $v'$ be the alias vertex associated with $v$. $H$ is a rectilinear planar representation of $G$ if and only if the following two conditions hold: The restriction $H_\eta$ of $H$ to $G_\eta$ is a rectilinear planar representation; the number $A$ of right turns minus left turns of any simple cycle of $G$ in $H$ containing $e$ and traversed clockwise in $H$ is equal to $4$. 
	We have $A=\sigma_\eta+\alpha_{u'}+\alpha_{v'}$, where: $\sigma_\eta$ is the spirality of $H_\eta$; for $w\in \{u',v'\}$, $\alpha_w=1$, $\alpha_w=0$, and $\alpha_w=-1$ if the angle formed by $w$ in $f_{int}$ is equal to $90^o$, $180^o$, or $270^o$, respectively. 
	
	According to the definition of root condition, there are three cases to consider: 
	$(i)$ $\Delta_\rho=[2,6]$, $(ii)$ $\Delta_\rho=[3,5]$, and $(iii)$ $\Delta_\rho=4$. 
	Consider Case $(i)$.  Since in this case the alias vertices coincide with the poles, we have $\alpha_{u'} \in [-1,1]$, $\alpha_{v'} \in [-1,1]$, and hence $\alpha_{u'}+\alpha_{v'} \in [-2,2]$.  
	If $G$ is rectilinear planar, we have that $A = \sigma_\eta+\alpha_{u'}+\alpha_{v'}=4$ for some $\sigma_\eta \in I_\eta$ and for $\alpha_{u'}+\alpha_{v'} \in [-2,2]$. Hence, $\sigma_\eta=4-\alpha_{u'}-\alpha_{v'} \in [2,6]$, i.e., the root condition $I_\eta \cap \Delta_\rho \neq \emptyset$ holds.    
	
	Suppose vice versa that the root condition $I_\eta \cap \Delta_\rho \neq \emptyset$ holds. For any value $\sigma_\eta \in I_\eta \cap \Delta_\rho$ there exists a rectilinear planar representation of $H_\eta$ of $G_\eta$ with spirality $\sigma_\eta$. Also, since $\Delta_\rho=[2,6]$, we have that  $4-\sigma_\eta\in [-2,2]$, and therefore, for any possible choice of $\sigma_\eta \in I_\eta \cap \Delta_\rho$, we can suitably choose $\alpha_{u'}$ and $\alpha_{v'}$ such that $\alpha_{u'}+\alpha_{v'}=4-\sigma_\eta$, i.e., $A = \sigma_\eta+\alpha_{u'}+\alpha_{v'}=4$. It follows that $G$ is rectilinear planar and it admits a rectilinear planar representation for any value $\sigma_\eta \in I_\eta \cap \Delta_\rho$. 
	
	Cases~$(ii)$ and~$(iii)$ can be proved analogously, observing that in Case~$(ii)$  $\alpha_{u'}+\alpha_{v'}\in [-1,1]$ and in Case~$(iii)$ $\alpha_{u'}+\alpha_{v'}=0$.  
\end{proof}

The next theorem is an immediate consequence of Lemma~\ref{le:root-condition} and Theorem~\ref{th:representability-conditions-intervals}, and it provides a characterization of rectilinear plane series-parallel graphs.

\begin{theorem}\label{th:rect-characterization}
	Let $G$ be a plane series-parallel graph and let $T$ be an SP$Q^*$-tree of $G$. Let $\eta$ be the root child of $T$. Graph $G$ is rectilinear planar if and only if: (i) $G_\eta$ is rectilinear planar; (ii) $G$ satisfies the root condition.
\end{theorem}

\renewcommand{\arraystretch}{1.5}
\begin{table}[tb]
	\centering
	\scriptsize
	\caption{Representability conditions and intervals for the different types of nodes. In the formulas $\gamma=\alpha+\beta-2$ and $\rho(\cdot)$ is such that $\rho(r)=1$ and $\rho(l)=0$.}\label{ta:representability}
	\begin{tabular}{|l |c |}
		\hline
		\rowcolor{antiquewhite}\multicolumn{2}{|c|}{\bf Q$^*$-node of length $\ell$ -- Lemma~\ref{le:Q*-representability}}\\
		\hline
		{Representability Condition} & true \\
		{Representability Interval} & $[-\ell+1, \ell-1]$\\
		
		\hline
		\rowcolor{antiquewhite}\multicolumn{2}{|c|}{\bf S-node with $h$ children -- Lemma~\ref{le:S-representability}}\\
		\hline
		{Representability Condition} & true \\
		{Representability Interval} & $[\sum_{i=1}^hm_i,\sum_{i=1}^hM_i]$\\
		
		\hline
		\rowcolor{antiquewhite}\multicolumn{2}{|c|}{\bf P-node with three children --  Lemma~\ref{le:P-3-children-representability}}\\
		\hline
		{Representability Condition} & $[m_l-2,M_l-2] \cap [m_c,M_c] \cap [m_r+2,M_r+2] \neq \emptyset$\\
		{Representability Interval} & $[\max\{m_l-2,m_c,m_r+2\},\min\{M_l-2,M_c,M_r+2\}]$\\
		
		\hline
		\rowcolor{antiquewhite}\multicolumn{2}{|c|}{\bf P-node with two children $-$ \Pio{2}{\alpha\beta} --  Lemma~\ref{le:P-2-children-representability-type1}}\\
		\hline
		{Representability Condition} & $[m_l-M_r, M_l-m_r] \cap [2, 4-\gamma] \neq \emptyset$ \\
		{Representability Interval} & $[\max\{m_l-2,m_r\}+\frac{\gamma}{2}, \min\{M_l, M_r+2\}-\frac{\gamma}{2}]$\\
		
		\hline
		\rowcolor{antiquewhite}\multicolumn{2}{|c|}{\bf P-node with two children $-$ \Pio{3d}{\alpha\beta} --  Lemma~\ref{le:P-2-children-representability-type2}}\\
		\hline
		{Representability Condition} & $[m_l-M_r, M_l-m_r] \cap [\frac{5}{2},\frac{7}{2}-\gamma] \neq \emptyset$ \\
		{Representability Interval} & $[\max\{m_l-\frac{3}{2},m_r+1\}+\frac{\gamma-\rho(d)}{2},\min\{M_l-\frac{1}{2}, M_r+2\}-\frac{\gamma+\rho(d)}{2}]$\\
		
		\hline
		\rowcolor{antiquewhite}\multicolumn{2}{|c|}{\bf P-node with two children $-$ \Pin{3dd'} --  Lemma~\ref{le:P-2-children-representability-type3}}\\
		\hline
		{Representability Condition} & $3 \in [m_l-M_r,M_l-m_r]$ \\
		{Representability Interval} & $[\max\{m_l-1,m_r+2\}-\frac{\rho(d)+\rho(d')}{2},\min\{M_l-1, M_r+2\}-\frac{\rho(d)+\rho(d')}{2}]$\\
		\hline
	\end{tabular}
\end{table}


\section{Rectilinear Planarity Testing Algorithm}\label{se:rect-alg-sp}

\begin{theorem}\label{th:testing}
	Let $G$ be an $n$-vertex plane series-parallel graph.
	There exists an $O(n)$-time algorithm that tests whether $G$ admits a planar rectilinear representation and that constructs one in the positive case.
\end{theorem}
\begin{proof}
	If $G$ is biconnected let $e$ be an edge of $G$ on the external face; otherwise, let $e$ be a dummy edge added on the external face to make $G$ biconnected. Let $T$ be an SPQ$^*$-tree of $G$ with respect to $e$. We first show how to perform the test in linear time. If the test is positive, we show how to efficiently construct a rectilinear planar representation of $G$.
	
	\smallskip\noindent{\sf Testing Algorithm.} Based on Theorem~\ref{th:representability-conditions-intervals}, the algorithm visits $T$ in post-order and, for each non-root node $\nu$ of $T$, it checks the representability condition of $\nu$ and computes interval $I_\nu$ if the condition is positive. If the representability condition is violated for some node, the algorithm halts and returns a negative answer. Otherwise, the algorithm reaches the root $\rho$ of $T$. If $e$ is a dummy edge, the algorithm halts and returns a positive answer (since $e$ will not appear in the representation, the algorithm does not need to check anything else). If $e$ is real, let $\eta$ be the child of $\rho$ other than the child associated with $e$ (see Fig.~\ref{fi:spq-tree}). Based on Theorem~\ref{th:rect-characterization}, to complete the test, the algorithm must check the root condition, i.e., it must check whether $I_\eta \cap \Delta_\rho \neq \emptyset$. 
	
	We now analyze the time complexity of the testing algorithm. $T$ can be computed in $O(n)$ time and it consists of $O(n)$ nodes~\cite{DBLP:books/ph/BattistaETT99}. For a node $\nu$ of $T$ that is not a Q$^*$-node, denote by $n_\nu$ the number of children of $\nu$. In the bottom-up visit, each node of $T$ is visited exactly once. By Theorem~\ref{th:representability-conditions-intervals}, for a non-root node $\nu$ of $T$ we have the following: If 
	$\nu$ is a Q$^*$-node, its representability interval can be computed in $O(1)$ time, assuming that the length $\ell$ of the chain of edges represented by $\nu$ is stored at $\nu$ during the construction of $T$. If $\nu$ is an S-node, its representability interval can be computed in $O(n_\nu)$ time. If $\nu$ is a P-node, its representability interval can be computed in $O(1)$ time. 
	Finally, the root condition is easily checked in $O(1)$ time by Lemma~\ref{le:root-condition}. 
	It follows that the whole test takes $O(n)$ time.
	
	\smallskip\noindent{\sf Construction Algorithm.} Suppose that the test is positive. By Theorem~\ref{th:rect-characterization}, the root condition holds, and by  Lemma~\ref{le:root-condition}, for  $\sigma_\eta \in I_\eta \cap \Delta_\rho$, $G$ admits  
	a rectilinear planar representation $H$ such that its restriction $H_\eta$ to $G_\eta$ has spirality $\sigma_\eta$. Hence, the algorithm starts by arbitrarily choosing a value $\sigma_\eta \in I_\eta \cap \Delta_\rho$ (if $e$ is a dummy edge, it can choose $\sigma_\eta$ as any value in $I_\eta$). Then, to construct a rectilinear planar representation $H$ of $G$, the algorithm visits $T$ top-down and determine the right value of spirality required by the component associated with each node of $T$ distinct from $\eta$. Once the spiralities for all nodes of $T$ are determined, $H$ is easily defined by fixing the vertex angles in each component as described in the proofs of  Lemmas~\ref{le:Q*-representability}--\ref{le:P-3-children-representability},~\ref{le:P-2-children-representability-type1}--\ref{le:P-2-children-representability-type3}. To compute the spiralities for the children of $\eta$ we distinguish the following cases:
	
	\smallskip\noindent{\bf Case~1:} $\eta$ is an S-node, with children $\mu_1, \dots, \mu_h$ $(i \in \{1, \dots, h\})$. Let $I_{\mu_i}=[m_i,M_i]$ be the representability interval of $\mu_i$. We must find a value $\sigma_{\mu_i} \in [m_i,M_i]$ for each $i = 1, \dots, h$ such that $\sum_{i=i}^h \sigma_{\mu_i} = \sigma_\eta$. To this aim, initially set $\sigma_{\mu_i} = M_i$ for each $i = 1, \dots, h$ and consider $s = (\sum_{i=i}^h \sigma_{\mu_i}) - \sigma_\eta$. By Lemma~\ref{le:spirality-S-node}, $s \geq 0$. If $s = 0$ we are done. Otherwise, iterate over all $i=1, \dots, h$ and for each $i$
	decrease both $\sigma_{\mu_i}$ and $s$ by the value $\min\{s, M_i - m_i\}$, until $s=0$.
	
	\smallskip\noindent{\bf Case~2:} $\eta$ is a P-node with three children, $\mu_l$, $\mu_c$, and $\mu_r$. By Lemma~\ref{le:spirality-P-node-3-children}, it suffices to set $\sigma_{\mu_l}=\sigma_{\eta}+2$, $\sigma_{\mu_c}=\sigma_{\eta}$, and $\sigma_{\mu_r}=\sigma_{\eta}-2$.
	
	\smallskip\noindent{\bf Case~3:} $\eta$ is a P-node with two children, $\mu_l$ and $\mu_r$. Let $u$ and $v$ be the poles of $\eta$. By Lemma~\ref{le:spirality-P-node-2-children}, $\sigma_{\mu_l}$ and  $\sigma_{\mu_r}$ must be fixed in such a way that $\sigma_{\mu_l} = \sigma_{\eta}+k_u^l\alpha_u^l+k_v^l\alpha_v^l$ and $\sigma_{\mu_r} = \sigma_{\eta}-k_u^r\alpha_u^r-k_v^r\alpha_v^r$. The values of $k_u^l$, $k_v^l$, $k_u^r$, and $k_v^r$ are fixed by the indegree and outdegree of $u$ and $v$. Hence, it suffices to choose the values of $\alpha_u^l$, $\alpha_v^l$, $\alpha_u^r$, $\alpha_v^r$ such that they are consistent with the type of $\eta$ and they yield $\sigma_{\mu_l} \in I_{\mu_l}$ and $\sigma_{\mu_r} \in I_{\mu_r}$. Since each $\alpha_w^d$ $(w \in \{u,v\}, d \in \{l,r\})$ is either $0$ or $1$ there are at most four possible combinations of values to consider.
	
	Once the spiralities for the children of $\eta$ are computed, the algorithm continues its top-down visit, and for each node $\nu$ for which a spirality $\sigma_\nu$ has been fixed, it computes the spiralities of the children of $\nu$ with same procedure as~for~$\eta$.
	Concerning the time complexity, the procedure in Case~1 takes linear time in the number of children of the S-node, while the procedures in Case~2 and Case~3 take constant time. Therefore the whole visit requires $O(n)$ time.
\end{proof}

Table~\ref{ta:runningexample} shows a running example based on Fig.~\ref{fi:prel}. For each P- and S-component it reports the representability interval computed in the bottom-up visit of the tree and the spirality fixed in the top-down visit (see also Fig.~\ref{fi:spiralities}).

\renewcommand{\arraystretch}{1.5}
\begin{table}[!h]
	\centering
	\caption{Running Example based on Figure~\ref{fi:prel}.}\label{ta:runningexample}
	\begin{tabular}{| c| c| c| c|}
		\hline
		\rowcolor{antiquewhite}\textsc{Node Label} &  \textsc{Node Type}
		& \textsc{Repres. Interval} &  \textsc{Spirality in $H$} \\
		\hline
		$\eta$ & S-node & $[-3,3]$& $3$  \\
		\hline
		$\nu_1$ & P-node (2 children) -- \Pio{3r}{11} &$[-2,2]$ & $2$ \\
		\hline
		$\nu_2$ & S-node &$[$-$4,4]$& $4$ \\
		\hline
		$\nu_3$ & S-node &$[$-$\frac{5}{2},\frac{1}{2}]$&$\frac{1}{2}$   \\
		\hline
		$\nu_4$ & P-node (2 children) -- \Pin{3lr} & $[0,0]$ & $0$  \\
		\hline
		$\nu_5$ & P-node (3 children) & $[$-$1,0]$& $0$  \\
		\hline
		$\nu_6$ & P-node (2 children) -- \Pio{2}{12} &  $[$-$\frac{3}{2},\frac{1}{2}]$&$\frac{1}{2}$  \\
		\hline
		$\nu_7$ & P-node (2 children) -- \Pio{2}{22} & $[0,0]$&$0$ \\
		\hline
		$\nu_8$ & S-node &$[$-$\frac{3}{2}, \frac{3}{2}]$ &$\frac{3}{2}$ \\
		\hline
		$\nu_9$ & S-node & $[$-$\frac{3}{2}, \frac{3}{2}]$ &-$\frac{3}{2}$ \\
		\hline
		$\nu_{10}$ & P-node (2 children) -- \Pio{2}{12} &$[$-$\frac{1}{2}, \frac{1}{2}]$ &$\frac{1}{2}$  \\
		\hline
		$\nu_{11}$ & P-node (2 children) -- \Pio{2}{12} &$[$-$\frac{1}{2}, \frac{1}{2}]$&-$\frac{1}{2}$  \\
		\hline
	\end{tabular}
\end{table}

\section{Conclusions and Open Problems}\label{se:open}

We proved that rectilinear planarity testing can be solved in linear time for series-parallel graphs with two terminals. Several open problems can be studied:

\smallskip\noindent \textsf{OP1.} Can we extend Theorem~\ref{th:testing} to 1-connected plane 4-graphs whose biconnected components are two-terminal series-parallel~graphs~(i.e., partial 2-trees)? The work in~\cite{frati-2020} solves the problem for 1-connected outerplanar graphs. 

\smallskip\noindent \textsf{OP2.} What is the time complexity of rectilinear planarity testing for general plane 4-graphs? The question is interesting even for triconnected plane 4-graphs.
A linear-time solution exists for plane 3-graphs ~\cite{DBLP:journals/jgaa/RahmanNN99,DBLP:conf/wg/RahmanN02}.

\smallskip\noindent \textsf{OP3.} Testing rectilinear planarity is NP-complete in the variable embedding setting but it can be solved in $O(n^3 \log n)$-time for series-parallel graphs~\cite{DBLP:journals/corr/abs-1908-05015}.  It is interesting to determine whether this complexity bound can be improved.



\bibliography{bibliography}
\bibliographystyle{abbrvurl}

\clearpage
\appendix
\makeatletter
\noindent
\rlap{\color[rgb]{0.51,0.50,0.52}\vrule\@width\textwidth\@height1\p@}%
\hspace*{7mm}\fboxsep1.5mm\colorbox[rgb]{1,1,1}{\raisebox{-0.4ex}{%
		\large\selectfont\sffamily\bfseries Appendix}}%
\makeatother

\section{Additional Material for Section~\ref{se:characterization}}

\subsection*{Proof of Lemma~\ref{le:P-2-children-representability-type1}}

We first prove the following result.

\itwotypeonesupport*
\begin{proof}
	We distinguish three cases, based on the values of $\alpha$ and $\beta$.
	
	\noindent{\bf Case 1: $\alpha=\beta=1$.}  In this case $G_\nu$ is of type \Pio{2}{11} and we prove that $G_\nu$ is rectilinear planar if and only if $G_{\mu_l}$ and $G_{\mu_r}$ are rectilinear planar for values of spiralities $\sigma_{\mu_l}$ and $\sigma_{\mu_r}$ such that $\sigma_{\mu_l}-\sigma_{\mu_r}\in[2,4]$. For a \Pio{2}{11} component we have $k_u^l=k_v^l=k_u^r=k_v^r=1$. 
	
	If $G_\nu$ is rectilinear planar, we have $1 \leq \alpha_{u}^l+\alpha_{u}^r \leq 2$ and $1 \leq \alpha_{v}^l+\alpha_{v}^r \leq 2$ in any rectilinear planar representation of $G_\nu$. Hence, by Lemma~\ref{le:spirality-P-node-2-children}, for any value of spirality $\sigma_\nu$  we have $\sigma_{\mu_l}-\sigma_{\mu_r}=\alpha_{u}^l+\alpha_{v}^l+\alpha_{u}^r+\alpha_{v}^r \in [2,4]$. 
	
	Suppose vice versa that $G_{\mu_l}$ and $G_{\mu_r}$ are rectilinear planar for values of spirality $\sigma_{\mu_l}$ and $\sigma_{\mu_r}$ such that $\sigma_{\mu_l}-\sigma_{\mu_r}\in[2,4]$. We show that $G_\nu$ admits a rectilinear planar representation $H_\nu$. To define $H_\nu$, we combine in parallel the two rectilinear planar representations of $G_{\mu_l}$ and $G_{\mu_r}$ and suitably assign the values of $\alpha_u^d$ and $\alpha_v^d$ ($d \in \{l,r\}$), depending on the value of $\sigma_{\mu_l}-\sigma_{\mu_r}$. This assignment is such that for any cycle $C$ of $G_\nu$ through $u$ and $v$, the number of $90^\circ$ angles minus the number of $270^\circ$ angles in the interior of $C$ is equal to four.
	Poles $u$ and $v$ split $C$ into two paths $\pi_l$ and $\pi_r$. The spirality $\sigma_{\mu_l}$ equals the number of right turns minus the number of left turns along $\pi_l$ while going from $u$ to $v$, which in turns corresponds to the number of $90^\circ$ angles minus the number of $270^\circ$ angles in the interior of $C$ at the vertices of $\pi_l$. Similarly, $-\sigma_{\mu_r}$ equals the number of right turns minus the number of left turns along $\pi_r$ while going from $v$ to $u$, which in turns corresponds to the number of $90^\circ$ angles minus the number of $270^\circ$ angles in the interior of $C$ at the vertices of $\pi_r$. By also taking into account the angles at $u$ and $v$ inside $C$, the number of $90^\circ$ angles minus the number of $270^\circ$ angles in the interior of $C$ can be expressed as $a_c = \sigma_{\mu_l}-\sigma_{\mu_r} + 4 - \alpha_u^l - \alpha_u^r - \alpha_v^l - \alpha_v^r$. 
	We distinguish the following three cases:
	%
	$(i)$ If $\sigma_{\mu_l}-\sigma_{\mu_r}=2$, then for every pole $w \in \{u,v\}$ we set  
	$\alpha_w^l$ and $\alpha_w^r$ such that $\alpha_w^l + \alpha_w^r = 1$. 
	$(ii)$ If $\sigma_{\mu_l}-\sigma_{\mu_r}=3$, then for one pole $w \in \{u,v\}$ we set $\alpha_w^l$ and $\alpha_w^r$ such that $\alpha_w^l + \alpha_w^r = 1$, and for the other pole $w' \in \{u,v\}$ we set $\alpha_{w'}^l = \alpha_{w'}^r = 1$.
	$(iii)$ If $\sigma_{\mu_l}-\sigma_{\mu_r}=4$, then for every pole $w \in \{u,v\}$ we set  
	$\alpha_w^l = \alpha_w^r = 1$.  
	In all the cases above, we have that $a_c=4$.
	Also, any other cycle not passing through $u$ and $v$ is an orthogonal polygon because it belongs to a rectilinear planar representation of either $G_{\mu_l}$ (with spirality $\sigma_{\mu_l}$) or $G_{\mu_r}$ (with spirality $\sigma_{\mu_r}$).
	
	\smallskip\noindent{\bf Case 2: $\alpha=1, \beta=2$.} In this case $G_\nu$ is of type \Pio{2}{12} 
	and we prove that $G_\nu$ is rectilinear planar if and only if $G_{\mu_l}$ and $G_{\mu_r}$ are rectilinear planar for values of spiralities $\sigma_{\mu_l}$ and $\sigma_{\mu_r}$ such that $\sigma_{\mu_l}-\sigma_{\mu_r}\in[2,3]$. Suppose, w.l.o.g., that $\outdeg_\nu(u)=1$ and $\outdeg_\nu(v)=2$. We have $k_u^l=k_u^r=1$ and $k_v^l=k_v^r=\frac{1}{2}$.
	
	If $G_\nu$ is rectilinear planar, we have $\alpha_v^l+\alpha_v^r=2$ and $\alpha_u^l+\alpha_u^r\in[1,2]$. By Lemma~\ref{le:spirality-P-node-2-children} we have $\sigma_{\mu_l}-\sigma_{\mu_r}=k_{u}^l \alpha_{u}^l +k_{v}^r \alpha_{v}^l +  k_{u}^l \alpha_{u}^r + k_{v}^r\alpha_{v}^r$, and hence $\sigma_{\mu_l}-\sigma_{\mu_r}= \alpha_{u}^l + \frac{1}{2}  \alpha_{v}^l + \alpha_{u}^r + \frac{1}{2} \alpha_{v}^r\in [2,3]$.
	
	Suppose vice versa that $G_{\mu_l}$ and $G_{\mu_r}$ are rectilinear planar for values of spiralities $\sigma_{\mu_l}$ and $\sigma_{\mu_r}$, such that $\sigma_{\mu_l}-\sigma_{\mu_r}\in [2,3]$. We show that $G_\nu$ admits a rectilinear planar representation $H_\nu$. 
	To define $H_\nu$, we combine in parallel the two rectilinear planar representations of $G_{\mu_l}$ and $G_{\mu_r}$ and suitably set $\alpha_u^d$ and $\alpha_v^d$ ($d \in \{l,r\}$). Namely, we set $\alpha_v^l=\alpha_v^r=1$. The values of $\alpha_u^l$ and $\alpha_u^r$ are set as follows: $(i)$ if $\sigma_{\mu_l}-\sigma_{\mu_r}=2$, we set $\alpha_u^l$ and $\alpha_u^r$ such that $\alpha_u^l + \alpha_u^r = 1$; $(ii)$ if $\sigma_{\mu_l}-\sigma_{\mu_r}=3$, we set $\alpha_u^l=\alpha_u^r=1$. With an argument similar to the previous case, for any cycle $C$ through $u$ and $v$, the number of $90^\circ$ angles minus the number of $270^\circ$ angles in the interior of $C$ can be expressed in this case by $a_c = \sigma_{\mu_l}-\sigma_{\mu_r} + 4 - \alpha_u^l - \alpha_u^r - 1$ (the angle at $v$ inside $C$ is always of 90$^\circ$ degrees). In case $(i)$ we have $a_c = 2 + 4 - 1 - 1 = 4$; in case $(ii)$ we have $a_c = 3 + 4 - 2 - 1 = 4$.
	Also, any other cycle not passing through $u$ and $v$ is an orthogonal polygon because it belongs to a rectilinear planar representation of either $G_{\mu_l}$ or $G_{\mu_r}$.
	
	\smallskip\noindent{\bf Case 3: $\alpha=\beta=2$.} In this case $G_\nu$ is of type \Pio{2}{22} and we prove that $G_\nu$ is rectilinear planar if and only if $G_{\mu_l}$ and $G_{\mu_r}$ are rectilinear planar for values of spiralities $\sigma_{\mu_l}$ and $\sigma_{\mu_r}$ such that $\sigma_{\mu_l}-\sigma_{\mu_r}=2$. We have $k_u^l=k_u^r=\frac{1}{2}$.
	
	If $G_\nu$ is rectilinear planar, we have that $\alpha_{u}^l+\alpha_{u}^r = \alpha_{v}^l+\alpha_{v}^r = 2$. By Lemma~\ref{le:spirality-P-node-2-children}, $\sigma_{\mu_l} = \sigma_{\nu} + 1$ and $\sigma_{\mu_r} = \sigma_{\nu} - 1$; hence $\sigma_{\mu_l}-\sigma_{\mu_r}=2$.  
	
	Suppose vice versa that $\sigma_{\mu_l}-\sigma_{\mu_r}=2$. We show that $G_\nu$ admits a rectilinear planar representation $H_\nu$. Again, we obtain $H_\nu$ by combining in parallel the two rectilinear planar representations of $G_{\mu_l}$ and $G_{\mu_r}$ and by suitably setting $\alpha_u^d$ and $\alpha_v^d$ ($d \in \{l,r\}$). In this case, for any cycle $C$ through $u$ and $v$, the number of $90^\circ$ angles minus the number of $270^\circ$ angles in the interior of $C$ can be expressed by $a_c = \sigma_{\mu_l}-\sigma_{\mu_r} + 1 + 1$ (both the angles at $u$ and $v$ inside $C$ is always of 90$^\circ$ degrees). We then set $\alpha_u^l=\alpha_v^l=\alpha_u^r=\alpha_v^r=1$, which guarantees $a_c = 4$.   
	Also, any other cycle not passing through $u$ and $v$ is an orthogonal polygon because it belongs to a rectilinear planar representation of either $G_{\mu_l}$ or $G_{\mu_r}$.
\end{proof}

\itwotypeone*
\begin{proof}
	We first prove the correctness of the representability condition and then the
	the validity of the representability interval.
	
	\smallskip\noindent\textsf{Representability condition.}
	Suppose that $G_\nu$ is rectilinear planar. By Lemma~\ref{le:P-2-children-support-type1}, there exist rectilinear planar representations for $G_{\mu_l}$ and $G_{\mu_r}$ with spiralities $\sigma_{\mu_l}$ and $\sigma_{\mu_r}$, respectively, such that $\sigma_{\mu_l}-\sigma_{\mu_r} \in [2,4-\gamma]$. Hence, $m_l-M_r \leq \sigma_{\mu_l}-\sigma_{\mu_r} \leq 4-\gamma$ and $M_l-m_r \geq \sigma_{\mu_l}-\sigma_{\mu_r} \geq 2$, i.e.,  $[m_l-M_r,M_l-m_r] \cap [2,4-\gamma] \neq \emptyset$.
	
	Suppose, vice versa that $[m_l-M_r,M_l-m_r] \cap [2,4-\gamma] \neq \emptyset$. By hypothesis $G_{\mu_l}$ (resp. $G_{\mu_r}$) is rectilinear planar for every integer value of spirality in the interval $[m_l,M_l]$ (resp. $[m_r,M_r]$). This implies that for every integer value $k$ in the interval $[m_l-M_r, M_l-m_r]$, there exist rectilinear planar representations for $G_{\mu_l}$ and $G_{\mu_r}$ with spiralities $\sigma_{\mu_l}$ and $\sigma_{\mu_r}$ such that $\sigma_{\mu_l}-\sigma_{\mu_r} = k$. Since by hypothesis there exists a value $k \in [m_l-M_r,M_l-m_r] \cap [2,4-\gamma]$, there must be two values of spiralities $\sigma_{\mu_l}$ and $\sigma_{\mu_r}$ for the representations of $G_{\mu_l}$ and $G_{\mu_r}$ such that $\sigma_{\mu_l}-\sigma_{\mu_r} = k \in [2,4-\gamma]$. Hence, by Lemma~\ref{le:P-2-children-support-type1} $G_\nu$ is rectilinear planar.
	
	\smallskip\noindent\textsf{Representability interval.}
	We analyze three cases, based on the values of $\alpha$ and $\beta$.
	
	\smallskip\noindent{\bf Case 1: $\alpha=\beta=1$.}  In this case $G_\nu$ is of type \Pio{2}{11} and we prove that $I_\nu = [\max\{m_l-2,m_r\}, \min\{M_l, M_r+2\}]$.  
	
	Assume first that $\sigma_\nu$ is the spirality of a rectilinear representation of $G_\nu$. By Lemma~\ref{le:spirality-P-node-2-children}, we have $\sigma_\nu \in [m_l-2, M_r+2]$. Also, since for a \Pio{2}{11} component we have $k_u^l=k_v^l=k_u^r=k_v^r=1$, we have $\sigma_\nu = \sigma_{\mu_r} + \alpha_{u}^r + \alpha_{v}^r$, which implies $\sigma_\nu \geq m_r$. Analogously, $\sigma_\nu = \sigma_{\mu_l} - \alpha_{u}^l - \alpha_{v}^l \leq M_l$. Hence, $\sigma_\nu \in I_\nu = [\max\{m_l-2,m_r\}, \min\{M_l,M_r+2\}]$.
	
	Assume vice versa that $k$ is any integer in the interval $I_\nu = [\max\{m_l-2,m_r\}, M = \min\{M_l,M_r+2\}]$. We show that $G_\nu$ admits a rectilinear planar representation with spirality $\sigma_\nu = k$. By hypothesis $k \leq \min\{M_l, M_r + 2\} \leq M_l$; also, $k \geq \max\{m_l-2,m_r\} \geq m_l-2$, i.e., $k+2 \geq m_l$. Hence $[k,k+2] \cap [m_l, M_l] \neq \emptyset$. Analogously, $k \leq \min\{M_l, M_r + 2\} \leq M_r + 2$, i.e., $k-2 \leq M_r$; also, $k \geq \max\{m_l-2,m_r\} \geq m_r$. Hence $[k-2,k] \cap [m_r,M_r] \neq \emptyset$.
	We now distinguish the following sub-cases:
	
	\begin{itemize}	
		\item{\bf Case 1.1: $k\le M_l-2$.} Consider any two rectilinear planar representations $H_{\mu_l}$ of $G_{\mu_l}$ and $H_{\mu_r}$ of $G_{\mu_r}$ with spirality $\sigma_{\mu_l}=k+2$ and $\sigma_{\mu_r}\in [k-2,k] \cap [m_r, M_r] \neq \emptyset$, respectively. Notice that, as already observed, $k+2 \geq m_l$ and by hypothesis $k+2 \le M_l$; hence $\sigma_{\mu_l} \in [m_l,M_l]$. With this choice we have $2\leq \sigma_{\mu_l}-\sigma_{\mu_r}\leq 4$, and we can combine $H_{\mu_l}$ and $H_{\mu_r}$ in parallel as in the proof of Lemma~\ref{le:P-2-children-support-type1} to obtain a rectilinear planar representation $H_\nu$ of $G_\nu$. By Lemma~\ref{le:spirality-P-node-2-children} the spirality of $H_\nu$ equals $\sigma_{\mu_l} - \alpha_{l}^u - \alpha_{l}^v=k+2 - \alpha_{l}^u - \alpha_{l}^v$ and it suffices to set $\alpha_{l}^u=\alpha_{l}^v=1$ (which is always possible, as these two values correspond to $90^\circ$ angles) to get $\sigma_\nu = k$.
		
		\item{\bf Case 1.2: $k= M_l-1$.} Consider any rectilinear planar representation $H_{\mu_l}$ of $G_{\mu_l}$ with spirality $\sigma_{\mu_l} = k+1 = M_l$. To suitably choose the spirality of a rectilinear planar representation $H_{\mu_r}$ of $G_{\mu_r}$, observe that by the representability condition $M_l-2 \geq m_r$ and, as already proved, $M_r \geq k-2$, i.e., $M_r \geq M_l-3$. It follows that $[M_l-3,M_l-2] \cap [m_r,M_r] \neq \emptyset$. Hence, either $M_l-3 \in [m_r,M_r]$ (possibly $m_r=M_r=M_l-3$) or $M_l-2 \in [m_r,M_r]$ (possibly $m_r=M_r=M_l-2$). In the first case, choose any representation $H_{\mu_r}$ with spirality $\sigma_r = M_l-3$, which implies $\sigma_{\mu_l}-\sigma_{\mu_r}=3 \in [2,4]$. In the second case, choose $H_{\mu_r}$ with spirality $\sigma_r = M_l-2$, which implies $\sigma_{\mu_l}-\sigma_{\mu_r}=2 \in [2,4]$. The two representations $H_{\mu_l}$ and $H_{\mu_r}$ can be combined in parallel to get a representation of $G_\nu$ with spirality $\sigma_{\nu}=k$. Namely, by Lemma~\ref{le:spirality-P-node-2-children} we can set $\alpha_u^l=0$ and $\alpha_v^l=1$ (or vice versa); also, if $\sigma_{\mu_r}=M_l-2$ we set $\alpha_u^r=0$ and $\alpha_v^l=1$ (or vice versa), while if $\sigma_{\mu_r}=M_l-3$ we set $\alpha_u^r=\alpha_v^l=1$.
		
		\item{\bf Case 1.3: $k= M_l$.} In this case, we can combine in parallel a representation $H_{\mu_l}$ of $G_{\mu_l}$ with spirality $\sigma_{\mu_l}=k=M_l$ and a representation  
		$H_{\mu_r}$ of $G_{\mu_r}$ with spirality $\sigma_{\mu_r}=k-2=M_l-2$, which implies that $ \sigma_{\mu_l}-\sigma_{\mu_r}=2$. By the representability condition we have $M_l-2 \geq m_r$, i.e., $\sigma_{\mu_r} \geq m_r$; also, $k \leq \min\{M_l,M_r+2\}\leq M_r+2$, i.e., $\sigma_{\mu_r} \leq M_r$. Hence, $\sigma_{\mu_r} \in [m_r,M_r]$. By Lemma~\ref{le:spirality-P-node-2-children} we can set $\alpha_u^l=\alpha_v^l=0$ and $\alpha_u^r=\alpha_v^l=1$ to get a representation of $G_\nu$ with spirality $\sigma_\nu = k$. 
	\end{itemize}
	
	\smallskip\noindent{\bf Case 2: $\alpha=1, \beta=2$.}  In this case $G_\nu$ is of type \Pio{2}{12} and we prove that $I_\nu = [\max\{m_l-2,m_r\}+\frac{1}{2}, \min\{M_l, M_r+2\}-\frac{1}{2}]$.
	
	Assume first that $G_\nu$ is rectilinear planar and let $H_\nu$ be a rectilinear planar representation of $G_\nu$ with spirality $\sigma_\nu$. Let $H_{\mu_l}$ and $H_{\mu_r}$ be the rectilinear planar representations of $G_{\mu_l}$ and $G_{\mu_r}$ contained in $H_\nu$, and let $\sigma_{\mu_l}$ and $\sigma_{\mu_r}$ be their corresponding spiralities. 
	By Lemma~\ref{le:P-2-children-support-type1}, $\sigma_{\mu_l} - \sigma_{\mu_r} \in [2,3]$, i.e., $\sigma_{\mu_l} \in [2+\sigma_{\mu_r}, 3+\sigma_{\mu_r}]$. Since $\sigma_{\mu_l} \in [m_l,M_l]$ and $\sigma_{\mu_r} \in [m_r,M_r]$, we have $\sigma_{\mu_l}\ge \max\{m_l,m_r+2\}$.
	
	Suppose, w.l.o.g, that $\outdeg_\nu(v)=2$ and $\outdeg_\nu(u)=1$. We have $k_{u}^r=k_{u}^l=1$,  $k_{v}^r=k_{v}^l=\frac{1}{2}$, $\alpha_{u}^l \in [0,1]$, $\alpha_{u}^r \in [0,1]$, and $\alpha_{v}^l=\alpha_{v}^r = 1$. 
	By Lemma~\ref{le:spirality-P-node-2-children} we have $\sigma_\nu = \sigma_{\mu_l} - \alpha_{u}^l - \frac{1}{2} \alpha_{v}^l$. Since $-\alpha_{u}^l - \frac{1}{2} \alpha_{v}^l\ge -\frac{3}{2}$, we have $\sigma_\nu\ge \max\{m_l,m_r+2\}-\frac{3}{2}$. It follows that $\sigma_\nu\ge \max\{m_l-2,m_r\}+\frac{1}{2}$. Analogously, since $\sigma_{\mu_r} \in [\sigma_{\mu_l}-3,\sigma_{\mu_l}-2]$, we have $\sigma_{\mu_r}\le \min\{M_l-2,M_r\}$. By Lemma~\ref{le:spirality-P-node-2-children} we have $\sigma_\nu = \sigma_{\mu_r} + \alpha_{u}^r + \frac{1}{2} \alpha_{v}^r$. Since $\alpha_{u}^r + \frac{1}{2} \alpha_{v}^r \le \frac{3}{2}$,  we have $\sigma_\nu\le \min\{M_l-2,M_r\}+\frac{3}{2}$. It follows that  $\sigma_\nu\le \min\{M_l,M_r+2\}-\frac{1}{2}$. Therefore, $\sigma_\nu \in I_\nu$.
	
	Assume vice versa that $k$ is a semi-integer in the interval $I_\nu = [\max\{m_l-2,m_r\}+\frac{1}{2},\min\{M_l,M_r+2\}-\frac{1}{2}]$. We show that $G_\nu$ has a rectilinear planar representation with spirality $\sigma_\nu = k$.
	Since $k \in [m_l-\frac{3}{2}, M_l-\frac{1}{2}]$ we have $k + \frac{1}{2} \leq M_l$ and $k + \frac{3}{2} \geq m_l$, i.e., $[k+\frac{1}{2},k+\frac{3}{2}]\cap [m_l,M_l] \neq \emptyset$. Also, since $m_l$ and $M_l$ are both integer numbers while $k$ is semi-integer, it is impossible to have $k + 1 = m_l = M_l$. It follows that $k+\frac{1}{2} \in [m_l,M_l]$ or $k+\frac{3}{2}\in [m_l,M_l]$. 
	With the same reasoning, we have $k \in [m_r+\frac{1}{2}, M_r+\frac{3}{2}]$ and $[k-\frac{3}{2},k-\frac{1}{2}] \cap [m_r,M_r] \neq \emptyset$. Hence, $k-\frac{3}{2}\in [m_r,M_r]$ or $k-\frac{1}{2}\in [m_r,M_r]$. We now prove that $k+\frac{3}{2} \in [m_l,M_l]$ or $k-\frac{3}{2} \in [m_r,M_r]$. Suppose by contradiction that $k+\frac{3}{2}\not \in [m_l,M_l]$ and $k-\frac{3}{2}\not \in [m_r,M_r]$. In that case $k+\frac{1}{2}\in [m_l,M_l]$ and $k-\frac{1}{2} \in [m_r,M_r]$. Consequently, $k+\frac{1}{2}=M_l$ and $k-\frac{1}{2}=m_r$. Hence, $M_l-m_r=1$ and, by the representability condition, $G_\nu$ is not rectilinear planar, a contradiction. 
	As in the previous case, a rectilinear representation of $G_\nu$ with spirality $k$ is obtained by combining in parallel a representations $H_{\mu_l}$ of $G_{\mu_l}$ with spirality $\sigma_{\mu_l}$ and a representation $H_{\mu_r}$ of $G_{\mu_r}$ with spirality $\sigma_{\mu_r}$, for two suitable values  $\sigma_{\mu_l}$ and $\sigma_{\mu_r}$. Based on the aforementioned analysis, we distinguish the following sub-cases:      
	
	\begin{itemize}
		\item \textbf{Case 2.1: $k+\frac{3}{2}\not \in [m_l,M_l]$.} 
		This implies that $k+\frac{1}{2}\in [m_l,M_l]$ and $k-\frac{3}{2}\in [m_r,M_r]$, and therefore we set $\sigma_{\mu_l}=k+\frac{1}{2}$ and $\sigma_{\mu_r}=k-\frac{3}{2}$.
		
		\item \textbf{Case 2.2: $k-\frac{3}{2}\not \in [m_r,M_r]$.} 
		This implies that $k+\frac{3}{2}\in [m_l,M_l]$ and $k-\frac{1}{2}\in [m_r,M_r]$, and therefore we set $\sigma_{\mu_l}=k+\frac{3}{2}$ and $\sigma_{\mu_r}=k-\frac{1}{2}$.
		
		\item \textbf{Case 2.3: $k+\frac{3}{2}\in [m_l,M_l]$ and $k-\frac{3}{2}\in [m_r,M_r]$.} We set $\sigma_{\mu_l}=k+\frac{3}{2}$ and $\sigma_{\mu_r}=k-\frac{3}{2}$.
	\end{itemize} 
	
	Notice that in all the three sub-cases we have $\sigma_{\mu_l}-\sigma_{\mu_r}\in [2,3]$, hence by Lemma~\ref{le:P-2-children-support-type1} there exists a rectilinear planar representation $H_\nu$ of $G_\nu$ that contains $H_{\mu_l}$ and $H_{\mu_r}$. It remains to prove that the spirality $\sigma_\nu$ of $H_\nu$ is equal to $k$. Suppose, w.l.o.g, that $\outdeg_\nu(u)=1$ and $\outdeg_\nu(v)=2$.  We have  $k_{u}^r=k_{u}^l=1$ and $k_{v}^r=k_{v}^l=\frac{1}{2}$. Since $G_\nu$ is rectilinear planar, we have $\alpha_{u}^l\in[0,1]$ and  $\alpha_{v}^l=1$. By Lemma~\ref{le:spirality-P-node-2-children} we have $\sigma_\nu = \sigma_{\mu_l} - \alpha_{u}^l - \frac{1}{2} \alpha_{v}^l$, where $\sigma_\nu$ is the spirality of the representation $H_\nu$ of $G_\nu$. In Case~2.1 we have $\sigma_\nu = k+\frac{1}{2} - \alpha_{u}^l - \frac{1}{2} \alpha_{v}^l$. By choosing $\alpha_{u}^l=0$ and $\alpha_{v}^l=1$ we have $\sigma_\nu = k$. In Case~2.2 and in Case~2.3  we have $\sigma_\nu = k+\frac{3}{2} - \alpha_{u}^l - \frac{1}{2} \alpha_{v}^l$. By choosing $\alpha_{u}^l=1$ and  $\alpha_{v}^l=1$ we have  $\sigma_\nu = k$.
	
	\smallskip\noindent{\bf Case 3: $\alpha=\beta=2$.}  In this case $G_\nu$ is of type \Pio{2}{22} and we prove that $I_\nu = [\max\{m_l-2,m_r\}+1, \min\{M_l, M_r+2\}-1]$. 
	
	Assume first that $G_\nu$ is rectilinear planar and let $H_\nu$ be a rectilinear planar representation of $G_\nu$ with spirality $\sigma_\nu$. Let $H_{\mu_l}$ and $H_{\mu_r}$ be the rectilinear planar representations of $G_{\mu_l}$ and $G_{\mu_r}$ contained in $H_\nu$, and let  $\sigma_{\mu_l}$ and $\sigma_{\mu_r}$ their spiralities. Since both $u$ and $v$ have outdegree two in $G_\nu$ we have that $\alpha_{u}^l+\alpha_{u}^r = \alpha_{v}^l+\alpha_{v}^r = 2$. By Lemma~\ref{le:spirality-P-node-2-children},  $\sigma_{\mu_l} = \sigma_{\nu} + 1$ and $\sigma_{\mu_r} = \sigma_{\nu} - 1$. By the representability condition $\sigma_{\mu_r}=\sigma_{\mu_l}-2$. Hence $\sigma_{\mu_r}\ge m_l-2$ and $\sigma_{\mu_r}\ge \max\{m_l-2,m_r\}$. 
	Also by $\sigma_{\nu}=\sigma_{\mu_r}+1$, $\sigma_\nu\ge \max\{m_l-2,m_r\}+1$. Similarly, by the representability condition $\sigma_{\mu_l}=\sigma_{\mu_r}+2$. Hence $\sigma_{\mu_l}\le M_r+2$ and $\sigma_{\mu_l}\le \max\{M_l,M_r+2\}$. Since $\sigma_{\mu_l}=\sigma_\nu+1$ we have $\sigma_\nu\le \max\{M_l,M_r+2\}-1$.
	
	Assume vice versa that $k$ is an integer in the interval $I_\nu = [\max\{m_l-2,m_r\}+1,\min\{M_l,M_r+2\}-1]$. We show that there exists a rectilinear planar representation of $G_\nu$ with spirality $\sigma_\nu=k$. We have  $k +1 \in [\max\{m_l,m_r+2\},\min\{M_l,M_r+2\}]$ and therefore $k +1 \in [m_l,M_l]$. Hence there exists a rectilinear planar representation $H_{\mu_l}$ of $G_{\mu_l}$ with spirality $\sigma_{\mu_l}=k+1$. Similarly, we have  $k -1 \in [\max\{m_l-2,m_r\},\min\{M_l-2,M_r\}]$ and therefore $k-1 \in [m_r,M_r]$. Hence there exists a rectilinear planar representation $H_{\mu_r}$ of $G_{\mu_r}$ with spirality $\sigma_{\mu_r} = k-1$. By the representability condition $G_\nu$ has a rectilinear planar representation $H_\nu$; also, following the same construction as in the proof of Lemma~\ref{le:P-2-children-support-type1}, the spirality of $H_\nu$ is $\sigma_\nu=k$.	
\end{proof}

\subsection*{Proof of Lemma~\ref{le:P-2-children-representability-type2}}

We first prove the following result.

\begin{restatable}{lemma}{itwotypetwosupport}\label{le:P-2-children-support-type2}
	Let $G_\nu$ be a P-node of type \Pio{3d}{\alpha\beta} and let $\mu_l$ and $\mu_r$ be its two children. $G_\nu$ is rectilinear planar if and only if $G_{\mu_l}$ and $G_{\mu_r}$ are rectilinear planar for values of spiralities $\sigma_{\mu_l}$ and $\sigma_{\mu_r}$, respectively, such that $\sigma_{\mu_l}-\sigma_{\mu_r}\in[\frac{5}{2},\frac{7}{2}-\gamma]$, where $\gamma = \alpha + \beta - 2$.
\end{restatable}
\begin{proof}
	We distinguish four cases, based on the values of $\alpha$, $\beta$, and $d$.
	
	\noindent{\bf Case 1: $\alpha=\beta=1$, $d=l$.}  In this case $G_\nu$ is of type \Pio{3l}{11} and we prove that $G_\nu$ is rectilinear planar if and only if $G_{\mu_l}$ and $G_{\mu_r}$ are rectilinear planar for values of spiralities $\sigma_{\mu_l}$ and $\sigma_{\mu_r}$ such that $\sigma_{\mu_l}-\sigma_{\mu_r}\in[\frac{5}{2},\frac{7}{2}]$. For a \Pio{3l}{11} component we have $k_u^l=k_u^r=k_v^r=1$ and $k_v^l=\frac{1}{2}$. 
	
	If $G_\nu$ is rectilinear planar, we have $1 \leq \alpha_{u}^l+\alpha_{u}^r \leq 2$ and $\alpha_{v}^l=\alpha_{v}^r=1$ in any rectilinear planar representation of $G_\nu$. Hence, by Lemma~\ref{le:spirality-P-node-2-children}, for any value of spirality $\sigma_\nu$  we have $\sigma_{\mu_l}-\sigma_{\mu_r}=\alpha_{u}^l+\frac{1}{2}\alpha_{v}^l+\alpha_{u}^r+\alpha_{v}^r \in [\frac{5}{2},\frac{7}{2}]$. 
	
	Suppose vice versa that $G_{\mu_l}$ and $G_{\mu_r}$ are rectilinear planar for values of spirality $\sigma_{\mu_l}$ and $\sigma_{\mu_r}$ such that $\sigma_{\mu_l}-\sigma_{\mu_r}\in[\frac{5}{2},\frac{7}{2}]$. We show that $G_\nu$ admits a rectilinear planar representation $H_\nu$. To define $H_\nu$, we combine in parallel the two rectilinear planar representations of $G_{\mu_l}$ and $G_{\mu_r}$ and suitably assign the values of $\alpha_u^l$ and $\alpha_u^r$, depending on the value of $\sigma_{\mu_l}-\sigma_{\mu_r}$. 
	
	Let $v'$ be the alias vertex of $G_{\mu_l}$ that is in $G_\nu$. Any cycle $C$ that goes through $u$ and $v$ also passes through $v'$. We show that the number of $90^\circ$ angles minus the number of $270^\circ$ angles in the interior of $C$ is equal to four.

	Vertices $u$ and $v'$ split $C$ into two paths $\pi_l$ and $\pi_r$. Suppose to visit $C$ clockwise. The number of right turns minus left turns along $\pi_l$ while going from $u$ to $v'$ equals $\sigma_{\mu_l}+\frac{1}{2}$.  The number of right turns minus left turns along $\pi_r$ while going from $v'$ to $u$ equals $-\sigma_{\mu_r}$. Hence, the sum $\sigma_{\mu_l}+\frac{1}{2}-\sigma_{\mu_r}+2-\alpha_u^r-\alpha_u^l$ corresponds to the number of $90^\circ$ angles minus the number of $270^\circ$ angles in the interior of $C$ at the vertices of $\pi_l$. Notice that $\alpha_u^r+\alpha_u^l\in \{1,2\}$ since $u$ is a vertex of degree 3. If $\sigma_{\mu_l}-\sigma_{\mu_r}=\frac{5}{2}$ we set $\alpha_u^r+\alpha_u^l=1$ and we have $\sigma_{\mu_l}+\frac{1}{2}-\sigma_{\mu_r}+2-\alpha_u^r-\alpha_u^l=\frac{5}{2}+\frac{1}{2}+2-1=4$. Else, if $\sigma_{\mu_l}-\sigma_{\mu_r}=\frac{7}{2}$ we set $\alpha_u^r+\alpha_u^l=2$ and we have $\sigma_{\mu_l}+\frac{1}{2}-\sigma_{\mu_r}+2-\alpha_u^r-\alpha_u^l=\frac{7}{2}+\frac{1}{2}+2-2=4$.

	Also, any other cycle not passing through $u$ and $v$ is an orthogonal polygon because it belongs to a rectilinear planar representation of either $G_{\mu_l}$ (with spirality $\sigma_{\mu_l}$) or $G_{\mu_r}$ (with spirality $\sigma_{\mu_r}$). 
	
	\noindent{\bf Case 2: $\alpha=1$, $\beta=2$, $d=l$.} In this case $G_\nu$ is of type \Pio{3l}{12} and we prove that $G_\nu$ is rectilinear planar if and only if $G_{\mu_l}$ and $G_{\mu_r}$ are rectilinear planar for values of spiralities $\sigma_{\mu_l}$ and $\sigma_{\mu_r}$ such that $\sigma_{\mu_l}-\sigma_{\mu_r}=\frac{5}{2}$ (note that this corresponds to the interval $[\frac{5}{2},\frac{7}{2}-\gamma]$ claimed in the lemma). For a \Pio{3l}{12} component we have $k_u^l=k_u^r=k_v^l=\frac{1}{2}$ and $k_v^r=1$.

	If $G_\nu$ is rectilinear planar, we have $ \alpha_{u}^l=\alpha_{u}^r=\alpha_{v}^l=\alpha_{v}^r=1$ in any rectilinear planar representation of $G_\nu$. Hence, by Lemma~\ref{le:spirality-P-node-2-children}, for any value of spirality $\sigma_\nu$  we have $\sigma_{\mu_l}-\sigma_{\mu_r}=\frac{1}{2}\alpha_{u}^l+\frac{1}{2}\alpha_{v}^l+\frac{1}{2}\alpha_{u}^r+\alpha_{v}^r =\frac{5}{2}$. 
	
	Suppose vice versa that $G_{\mu_l}$ and $G_{\mu_r}$ are rectilinear planar for values of spirality $\sigma_{\mu_l}$ and $\sigma_{\mu_r}$ such that $\sigma_{\mu_l}-\sigma_{\mu_r}=\frac{5}{2}$. We show that $G_\nu$ admits a rectilinear planar representation $H_\nu$. To define $H_\nu$, we combine in parallel the two rectilinear planar representations of $G_{\mu_l}$ and $G_{\mu_r}$ and assign values  $\alpha_u^l=\alpha_v^l=\alpha_u^r=\alpha_v^r=1$. 
	Let $v'$ be the alias vertex of $G_{\mu_l}$ that is in $G_\nu$. Any cycle $C$ that goes through $u$ and $v$ also passes through $v'$. We show that the number of $90^\circ$ angles minus the number of $270^\circ$ angles in the interior of $C$ is equal to four.
	
	Vertices $u$ and $v'$ split $C$ into two paths $\pi_l$ and $\pi_r$. Suppose to visit $C$ clockwise. The number of right turns minus left turns along $\pi_l$ while going from $u$ to $v'$ equals $\sigma_{\mu_l}+\frac{1}{2}$.  The number of right turns minus left turns along $\pi_r$ while going from $v'$ to $u$ equals $-\sigma_{\mu_r}$. Also, pole $u$ forms a $90^\circ$ angle inside $C$. Hence, the sum $\sigma_{\mu_l}+\frac{1}{2}-\sigma_{\mu_r}+1$ corresponds to the number of $90^\circ$ angles minus the number of $270^\circ$ angles in the interior of $C$ at the vertices of $\pi_l$. Since $\sigma_{\mu_l}-\sigma_{\mu_r}=\frac{5}{2}$ we have $\sigma_{\mu_l}+\frac{1}{2}-\sigma_{\mu_r}+1=\frac{5}{2}+\frac{1}{2}+1=4$.
	
	Also, any other cycle not passing through $u$ and $v$ is an orthogonal polygon because it belongs to a rectilinear planar representation of either $G_{\mu_l}$ (with spirality $\sigma_{\mu_l}$) or $G_{\mu_r}$ (with spirality $\sigma_{\mu_r}$).

	\noindent{\bf Case 3: $\alpha=\beta=1$, $d=r$.} This case is symmetric to Case 1.
	
	\noindent{\bf Case 4: $\alpha=1$, $\beta=2$, $d=r$.} This case is symmetric to Case 2. 
\end{proof}

\itwotypetwo*
\begin{proof}
	We first prove the correctness of the representability condition and then the
	the validity of the representability interval.
	
	\smallskip\noindent\textsf{Representability condition.}
	Suppose that $G_\nu$ is rectilinear planar. By Lemma~\ref{le:P-2-children-support-type2}, there exist rectilinear planar representations for $G_{\mu_l}$ and $G_{\mu_r}$ with spiralities $\sigma_{\mu_l}$ and $\sigma_{\mu_r}$, respectively, such that $\sigma_{\mu_l}-\sigma_{\mu_r} \in [\frac{5}{2},\frac{7}{2}-\gamma]$, where $\gamma = \alpha + \beta -2$. Hence, $m_l-M_r \leq \sigma_{\mu_l}-\sigma_{\mu_r} \leq \frac{7}{2}-\gamma$ and $M_l-m_r \geq \sigma_{\mu_l}-\sigma_{\mu_r} \geq \frac{5}{2}$, i.e.,  $[m_l-M_r,M_l-m_r] \cap [\frac{5}{2},\frac{7}{2}-\gamma] \neq \emptyset$.
	
	Suppose, vice versa that $[m_l-M_r,M_l-m_r] \cap [\frac{5}{2},\frac{7}{2}-\gamma] \neq \emptyset$. By hypothesis $G_{\mu_l}$ (resp. $G_{\mu_r}$) is rectilinear planar for every value of spirality in the interval $[m_l,M_l]$ (resp. $[m_r,M_r]$). This implies that for every semi-integer value $k$ in the interval $[m_l-M_r, M_l-m_r]$, there exist rectilinear planar representations for $G_{\mu_l}$ and $G_{\mu_r}$ with spiralities $\sigma_{\mu_l}$ and $\sigma_{\mu_r}$ such that $\sigma_{\mu_l}-\sigma_{\mu_r} = k$. Since by hypothesis there exists a value $k \in [m_l-M_r,M_l-m_r] \cap [\frac{5}{2},\frac{7}{2}-\gamma]$, there must be two values of spiralities $\sigma_{\mu_l}$ and $\sigma_{\mu_r}$ for the representations of $G_{\mu_l}$ and $G_{\mu_r}$ such that $\sigma_{\mu_l}-\sigma_{\mu_r} = k \in [\frac{5}{2},\frac{7}{2}-\gamma]$. Hence, by Lemma~\ref{le:P-2-children-support-type2} $G_\nu$ is rectilinear planar.
	
	\smallskip\noindent\textsf{Representability interval.} 
	We analyze four cases, based on the values of $\alpha$, $\beta$, and $d$ and we assume, w.l.o.g., that $v$ is the pole of degree four.
	
	\noindent{\bf Case 1: $\alpha=\beta=1$, $d=l$.}  In this case $G_\nu$ is of type \Pio{3l}{11} and we prove that $I_\nu = [\max\{m_l-\frac{3}{2},m_r+1\},\min\{M_l-\frac{1}{2}, M_r+2\}]$.
	
	Assume first that $\sigma_\nu$ is the spirality of a rectilinear planar representation of $G_\nu$.  Since for a \Pio{3l}{11} component we have $k_u^l=k_u^r=k_v^r=1$ and $k_v^l=\frac{1}{2}$, by Lemma~\ref{le:spirality-P-node-2-children} we 
	have $\sigma_\nu = \sigma_{\mu_r} + \alpha_{u}^r + \alpha_{v}^r$ and $\sigma_\nu = \sigma_{\mu_l} - \alpha_{u}^l - \frac{1}{2}\alpha_{v}^l$. Since $\alpha_u^l + \alpha_u^r \in \{1,2\}$ and $\alpha_v^l=\alpha_v^r=1$, we have:
	$\sigma_\nu \geq m_r + 1$, $\sigma_\nu \leq M_r+2$, $\sigma_\nu \geq m_l-\frac{3}{2}$, and $\sigma_\nu \leq M_l-\frac{1}{2}$.

	We now show that, if $\sigma_\nu \in  [\max\{m_l-\frac{3}{2},m_r+1\},\min\{M_l-\frac{1}{2},M_r+2\}]$, there exists a rectilinear planar representation of $G_\nu$ with spirality $\sigma_\nu$. We have $\sigma_\nu \in [m_l-\frac{3}{2}, M_l-\frac{1}{2}]$.
	Hence, $\sigma_\nu + \frac{1}{2} \leq M_l$ and $\sigma_\nu + \frac{3}{2} \geq m_l$, i.e., $[\sigma_\nu+\frac{1}{2},\sigma_{\nu}+\frac{3}{2}]\cap [m_l,M_l]\not = \emptyset$. Also, since $m_l$ and $M_l$ are both semi-integer numbers while $\sigma_\nu$ is integer, it is impossible to have $\sigma_\nu + 1 = m_l = M_l$. It follows that $\sigma_\nu+\frac{1}{2}\in [m_l,M_l]$ or $\sigma_\nu+\frac{3}{2}\in [m_l,M_l]$. With the same reasoning, we have $\sigma_\nu \in [m_r+1, M_r+2]$ and $[\sigma_\nu-2,\sigma_{\nu}-1]\cap [m_r,M_r]\not = \emptyset$. Hence, $\sigma_\nu-2\in [m_r,M_r]$ or $\sigma_\nu-1\in [m_r,M_r]$. 
	We now prove the following.
	
	\begin{claim}\label{cl:either3/2and2-L}
		Either $\sigma_\nu+\frac{3}{2}\in [m_l,M_l]$ or $\sigma_\nu-2 \in [m_r,M_r]$.
	\end{claim}
	\begin{proof}
		Suppose by contradiction that $\sigma_\nu+\frac{3}{2}\not \in [m_l,M_l]$ and $\sigma_\nu-2\not \in [m_r,M_r]$. In that case $\sigma_\nu+\frac{1}{2}\in [m_l,M_l]$ and $\sigma_\nu-1 \in [m_r,M_r]$. Consequently, $\sigma_\nu+\frac{1}{2}=M_l$ and $\sigma_\nu-1=m_r$. Hence, $M_l-m_r=\frac{3}{2}$ and,  by the representability condition, $G_\nu$ is not rectilinear planar. This is a contradiction. Hence, either $\sigma_\nu+\frac{3}{2}\in [m_l,M_l]$ or $\sigma_\nu-2 \in [m_r,M_r]$. 
	\end{proof}
	
	We can construct a rectilinear planar representation $H_{\mu_l}$ of $G_{\mu_l}$ with spirality $\sigma_{\mu_l}$ and a rectilinear planar representation $H_{\mu_r}$ of $G_{\mu_r}$ with spirality $\sigma_{\mu_r}$, based on the following cases:
	\begin{itemize}
		\item \textbf{Case (a):} $\sigma_\nu+\frac{3}{2}\not \in [m_l,M_l]$. 
		This implies that $\sigma_\nu+\frac{1}{2}\in [m_l,M_l]$ and $\sigma_\nu-2\in [m_r,M_r]$, and therefore we set 
		$\sigma_{\mu_l}=\sigma_\nu+\frac{1}{2}$ and $\sigma_{\mu_r}=\sigma_\nu-2$.
		
		\item \textbf{Case (b):} $\sigma_\nu-2\not \in [m_r,M_r]$. 
		This implies that $\sigma_\nu+\frac{3}{2}\in [m_l,M_l]$ and $\sigma_\nu-1\in [m_r,M_r]$, and therefore we set 
		$\sigma_{\mu_l}=\sigma_\nu+\frac{3}{2}$ and $\sigma_{\mu_r}=\sigma_\nu-1$.
		
		\item \textbf{Case (c):} $\sigma_\nu+\frac{3}{2}\in [m_l,M_l]$ and $\sigma_\nu-2\in [m_r,M_r]$.
		We set $\sigma_{\mu_l}=\sigma_\nu+\frac{3}{2}$ and $\sigma_{\mu_r}=\sigma_\nu-2$.
	\end{itemize} 
	
	By the claim proved above either the condition of Case~(a), Case~(b), or Case~(c) is verified. Notice that in all the three cases we have $\sigma_{\mu_l}-\sigma_{\mu_r}\in [\frac{5}{2},\frac{7}{2}]$, hence, there exists a rectilinear planar representation $H_\nu$ of $G_\nu$ given the values of $\sigma_{\mu_l}$ and $\sigma_{\mu_r}$ described in the three cases. We have to prove that in the three cases the spirality of $H_\nu$ is $\sigma_\nu$. 
	By Lemma~\ref{le:spirality-P-node-2-children} we have $\sigma_\nu' = \sigma_{\mu_l} - \alpha_{u}^l - \frac{1}{2} \alpha_{v}^l$, where $\sigma_\nu'$ is the spirality of the representation $H_\nu$ of $G_\nu$ given a choice of $\sigma_{\mu_l}$, $\alpha_{v}^l$, and $\alpha_{v}^r$. In Case~(a) we have $\sigma_\nu' = \sigma_\nu+\frac{1}{2} - \alpha_{u}^l - \frac{1}{2} \alpha_{v}^l$. By choosing $\alpha_{u}^l=0$ and  $\alpha_{v}^l=1$ we have  $\sigma_\nu' = \sigma_\nu$. In Case~(b) and in Case~(c)  we have $\sigma_\nu' = \sigma_\nu+\frac{3}{2} - \alpha_{u}^l - \frac{1}{2} \alpha_{v}^l$. By choosing $\alpha_{u}^l=1$ and  $\alpha_{v}^l=1$ we have  $\sigma_\nu' = \sigma_\nu$.
	
	\smallskip
	\noindent{\bf Case 2: $\alpha=1$, $\beta=2$, $d=l$.}  In this case $G_\nu$ is of type \Pio{3l}{12} and we prove that $I_\nu = [\max\{m_l-\frac{3}{2},m_r+1\}+\frac{1}{2},\min\{M_l-\frac{1}{2}, M_r+2\}-\frac{1}{2}] = [\max\{m_l-1,m_r+\frac{3}{2}\},\min\{M_l-1, M_r+\frac{3}{2}\}$.
	
	Assume first that $\sigma_\nu$ is the spirality of a rectilinear planar representation of $G_\nu$.  Since for a \Pio{3l}{12} component we have $k_u^l=k_u^r=k_v^l=\frac{1}{2}$ and $k_v^r=1$, by Lemma~\ref{le:spirality-P-node-2-children} we 
	have $\sigma_\nu = \sigma_{\mu_r} + \frac{1}{2}\alpha_{u}^r + \alpha_{v}^r$ and $\sigma_\nu = \sigma_{\mu_l} - \frac{1}{2}\alpha_{u}^l - \frac{1}{2}\alpha_{v}^l$. Since $\alpha_u^l=\alpha_v^l=\alpha_u^r=\alpha_v^r=1$, we have:
	$\sigma_\nu \geq m_r + \frac{3}{2}$, $\sigma_\nu \leq M_r+\frac{3}{2}$, $\sigma_\nu \geq m_l-1$, and $\sigma_\nu \leq M_l-1$.

	We now show that, if $\sigma_\nu \in [\max\{m_l-1,m_r+\frac{3}{2}\},\min\{M_l-1,M_r+\frac{3}{2}\}]$, there exists a rectilinear planar representation of $G_\nu$ with spirality $\sigma_\nu$. 	We have $\sigma_\nu \in [m_l-1, M_l-1]$ and $\sigma_\nu \in [m_r+\frac{3}{2}, M_r+\frac{3}{2}]$. Hence, $\sigma_\nu+1 \in [m_l, M_l]$ and $\sigma_\nu-\frac{3}{2} \in [m_r, M_r]$.  We can construct a rectilinear planar representation $H_{\mu_l}$ of $G_{\mu_l}$ with spirality $\sigma_{\mu_l}=\sigma_\nu+1$ and a rectilinear planar representation $H_{\mu_r}$ of $G_{\mu_r}$ with spirality $\sigma_{\mu_r}=\sigma_\nu-\frac{3}{2}$. Notice that, for this choice, we have $\sigma_{\mu_l}-\sigma_{\mu_r}=\frac{5}{2}$, hence, there exists a rectilinear planar representation $H_\nu$ of $G_\nu$ given the values of $\sigma_{\mu_l}$ and $\sigma_{\mu_r}$. We have to prove that the spirality of $H_\nu$ is $\sigma_\nu$.
	By Lemma~\ref{le:spirality-P-node-2-children} we have $\sigma_\nu' = \sigma_{\mu_l} - \frac{1}{2}\alpha_{u}^l - \frac{1}{2}\alpha_{v}^l$, where $\sigma_\nu'$ is the spirality of the representation $H_\nu$ of $G_\nu$ given a choice of $\sigma_{\mu_l}$, $\alpha_{v}^l$, and $\alpha_{v}^r$. Since $\sigma_\nu=\sigma_{\mu_l}-1$, $\alpha_{v}^l=1$, and $\alpha_{v}^r=1$, we have  $\sigma_\nu' = \sigma_\nu+1-1=\sigma_\nu$.
	
	\smallskip\noindent{\bf Case 3: $\alpha=\beta=1$, $d=r$}. This case is symmetric to Case~1.
	
	\smallskip\noindent{\bf Case 4: $\alpha=1$, $\beta=2$, $d=r$}. This case is symmetric to Case~2.  
	
\end{proof}

\subsection*{Proof of Lemma~\ref{le:P-2-children-representability-type3}}

We first prove the following result.

\begin{restatable}{lemma}{itwotypethreesupport}\label{le:P-2-children-support-type3}
	Let $G_\nu$ be a P-node of type \Pin{3dd'} and let $\mu_l$ and $\mu_r$ be its two children. 
	$G_\nu$ is rectilinear planar if and only if $G_{\mu_l}$ and $G_{\mu_r}$ are rectilinear planar for values of spiralities $\sigma_{\mu_l}$ and $\sigma_{\mu_r}$, respectively, such that $\sigma_{\mu_l}-\sigma_{\mu_r}=3$.
\end{restatable}
\begin{proof}
	We distinguish three cases, based on the values of $d$ and $d'$. Note that the proof for \Pin{3rl} is symmetric to the proof for \Pin{3lr}.
	
	\noindent{\bf Case 1: $d=d'=l$.}  In this case $G_\nu$ is of type \Pin{3ll} and we prove that $G_\nu$ is rectilinear planar if and only if $G_{\mu_l}$ and $G_{\mu_r}$ are rectilinear planar for values of spiralities $\sigma_{\mu_l}$ and $\sigma_{\mu_r}$ such that $\sigma_{\mu_l}-\sigma_{\mu_r}=3$. For a \Pin{3ll} component we have $k_u^l=k_v^l=\frac{1}{2}$ and $k_u^r=k_v^r=1$. 
	
	If $G_\nu$ is rectilinear planar, we have $\alpha_{u}^l=\alpha_{u}^r =\alpha_{v}^l=\alpha_{v}^r=1$ in any rectilinear planar representation of $G_\nu$. Hence, by Lemma~\ref{le:spirality-P-node-2-children}, for any value of spirality $\sigma_\nu$  we have $\sigma_{\mu_l}-\sigma_{\mu_r}=\frac{1}{2}\alpha_{u}^l+\frac{1}{2}\alpha_{v}^l+\alpha_{u}^r+\alpha_{v}^r=3$. 
	
	Suppose vice versa that $G_{\mu_l}$ and $G_{\mu_r}$ are rectilinear planar for values of spirality $\sigma_{\mu_l}$ and $\sigma_{\mu_r}$ such that $\sigma_{\mu_l}-\sigma_{\mu_r}=3$. We show that $G_\nu$ admits a rectilinear planar representation $H_\nu$. To define $H_\nu$, we combine in parallel the two rectilinear planar representations of $G_{\mu_l}$ and $G_{\mu_r}$ and assign values  $\alpha_u^l=\alpha_v^l=\alpha_u^r=\alpha_v^r=1$. 
	Let $u'$ and $v'$ be the alias vertices of $G_{\mu_l}$ that are in $G_\nu$. Any cycle $C$ that goes through $u$ and $v$ also passes through $u'$ and $v'$. We show that the number of $90^\circ$ angles minus the number of $270^\circ$ angles in the interior of $C$ is equal to four.
	
	Vertices $u'$ and $v'$ split $C$ into two paths $\pi_l$ and $\pi_r$. Suppose to visit $C$ clockwise. The number of right turns minus left turns along $\pi_l$ while going from $u'$ to $v'$ equals $\sigma_{\mu_l}+1$.  The number of right turns minus left turns along $\pi_r$ while going from $v'$ to $u'$ equals $-\sigma_{\mu_r}$. The sum of these two values corresponds to the number of $90^\circ$ angles minus the number of $270^\circ$ angles in the interior of $C$ at the vertices of $\pi_l$. Hence, $\sigma_{\mu_l}+1-\sigma_{\mu_r}=3+1=4$.
	
	Also, any other cycle not passing through $u$ and $v$ is an orthogonal polygon because it belongs to a rectilinear planar representation of either $G_{\mu_l}$ (with spirality $\sigma_{\mu_l}$) or $G_{\mu_r}$ (with spirality $\sigma_{\mu_r}$). 
	
	\noindent{\bf Case 2: $d=d'=r$.} This case is symmetric to Case 1, observing that $k_u^r=k_v^r=\frac{1}{2}$ and $k_u^l=k_v^l=1$.
	
	\noindent{\bf Case 3: $d=l$, $d'=r$.} In this case $G_\nu$ is of type \Pin{3lr} and we prove that $G_\nu$ is rectilinear planar if and only if $G_{\mu_l}$ and $G_{\mu_r}$ are rectilinear planar for values of spiralities $\sigma_{\mu_l}$ and $\sigma_{\mu_r}$ such that $\sigma_{\mu_l}-\sigma_{\mu_r}=3$. For a \Pin{3lr} component we have $k_u^r=k_v^l=\frac{1}{2}$ and $k_u^l=k_v^r=1$. 
	
	If $G_\nu$ is rectilinear planar, we have $\alpha_{u}^l=\alpha_{u}^r =\alpha_{v}^l=\alpha_{v}^r=1$ in any rectilinear planar representation of $G_\nu$. Hence, by Lemma~\ref{le:spirality-P-node-2-children}, for any value of spirality $\sigma_\nu$  we have $\sigma_{\mu_l}-\sigma_{\mu_r}=\alpha_{u}^l+\frac{1}{2}\alpha_{v}^l+\frac{1}{2}\alpha_{u}^r+\alpha_{v}^r=3$.
	
	Suppose vice versa that $G_{\mu_l}$ and $G_{\mu_r}$ are rectilinear planar for values of spirality $\sigma_{\mu_l}$ and $\sigma_{\mu_r}$ such that $\sigma_{\mu_l}-\sigma_{\mu_r}=3$. We show that $G_\nu$ admits a rectilinear planar representation $H_\nu$. To define $H_\nu$, we combine in parallel the two rectilinear planar representations of $G_{\mu_l}$ and $G_{\mu_r}$ and assign values  $\alpha_u^l=\alpha_v^l=\alpha_u^r=\alpha_v^r=1$. 
	%
	Let $v'$ be the alias vertex of the pole $v$ of $G_{\mu_l}$ such that $v'$ is along an edge of $G_\nu$. Similarly, let $u'$ be the alias vertex of the pole $u$ of $G_{\mu_r}$ such that $u'$ is along an edge of $G_\nu$. Any cycle $C$ that goes through $u$ and $v$ also passes through $u'$ and $v'$. We show that the number of $90^\circ$ angles minus the number of $270^\circ$ angles in the interior of $C$ is equal to four.
	
	Vertices $u'$ and $v'$ split $C$ into two paths $\pi_l$ and $\pi_r$. Suppose to visit $C$ clockwise. The number of right turns minus left turns along $\pi_l$ while going from $u'$ to $v'$ equals $\sigma_{\mu_l}+\frac{1}{2}$.  The number of right turns minus left turns along $\pi_r$ while going from $v'$ to $u'$ equals $-\sigma_{\mu_r}+\frac{1}{2}$. The sum of these two values corresponds to the number of $90^\circ$ angles minus the number of $270^\circ$ angles in the interior of $C$ at the vertices of $\pi_l$. Hence, $\sigma_{\mu_l}+\frac{1}{2}-\sigma_{\mu_r}+\frac{1}{2}=3+\frac{1}{2}+\frac{1}{2}=4$.
	
	Also, any other cycle not passing through $u$ and $v$ is an orthogonal polygon because it belongs to a rectilinear planar representation of either $G_{\mu_l}$ (with spirality $\sigma_{\mu_l}$) or $G_{\mu_r}$ (with spirality $\sigma_{\mu_r}$). 
\end{proof}

\itwotypethree*
\begin{proof}
	We first prove the correctness of the representability condition and then the
	the validity of the representability interval.
	
	\smallskip\noindent\textsf{Representability condition.}
	Suppose that $G_\nu$ is rectilinear planar. By Lemma~\ref{le:P-2-children-support-type3}, there exist rectilinear planar representations for $G_{\mu_l}$ and $G_{\mu_r}$ with spiralities $\sigma_{\mu_l}$ and $\sigma_{\mu_r}$, respectively, such that $\sigma_{\mu_l}-\sigma_{\mu_r} = 3$. Hence, $m_l-M_r \leq \sigma_{\mu_l}-\sigma_{\mu_r} \leq 3$ and $M_l-m_r \geq \sigma_{\mu_l}-\sigma_{\mu_r} \geq 3$, i.e.,  $3 \in [m_l-M_r,M_l-m_r]$.
	
	Suppose, vice versa that $3 \in [m_l-M_r,M_l-m_r]$. By hypothesis $G_{\mu_l}$ (resp. $G_{\mu_r}$) is rectilinear planar for every value of spirality in the interval $[m_l,M_l]$ (resp. $[m_r,M_r]$). This implies that there exist rectilinear planar representations for $G_{\mu_l}$ and $G_{\mu_r}$ with spiralities $\sigma_{\mu_l} \in [m_l,M_l]$ and $\sigma_{\mu_r} \in [m_r,M_r]$ such that $\sigma_{\mu_l}-\sigma_{\mu_r} = 3$. Hence, by Lemma~\ref{le:P-2-children-support-type3} $G_\nu$ is rectilinear planar.
	
	\smallskip\noindent\textsf{Representability interval.} We distinguish three cases, based on the values of $d$ and $d'$. Note that a possible forth case for \Pin{3rl} is symmetric to the case \Pin{3lr}. 
	\noindent{\bf Case 1: $d=d'=l$.} In this case $G_\nu$ is of type \Pin{3ll} and we prove that $I_\nu = [\max\{m_l-1,m_r+2\},\min\{M_l-1, M_r+2\}]$.
	
	Assume first that $\sigma_\nu$ is the spirality of a rectilinear planar representation of $G_\nu$.  Since for a \Pin{3ll} component we have $k_u^l=k_v^l=\frac{1}{2}$ and $k_u^r=k_v^r=1$, by Lemma~\ref{le:spirality-P-node-2-children} we 
	have $\sigma_\nu = \sigma_{\mu_r} + \alpha_{u}^r + \alpha_{v}^r$ and $\sigma_\nu =\sigma_{\mu_l} - \frac{1}{2}\alpha_{u}^l - \frac{1}{2}\alpha_{v}^l$. Since $\alpha_u^l=\alpha_v^l=\alpha_u^r=\alpha_v^r=1$, we have:
	$\sigma_\nu \geq m_r + 2$, $\sigma_\nu \leq M_r+2$, $\sigma_\nu \geq m_l-1$, and $\sigma_\nu \leq M_l-1$. 
	
	We now show that, if $\sigma_\nu \in [\max\{m_l-1,m_r+2\},\min\{M_l-1,M_r+2\}]$, there exists a rectilinear planar representation of $G_\nu$ with spirality $\sigma_\nu$. 	We have $\sigma_\nu \in [m_l-1, M_l-1]$ and $\sigma_\nu \in [m_r+2, M_r+2]$. Hence, $\sigma_\nu+1 \in [m_l, M_l]$ and $\sigma_\nu-2 \in [m_r, M_r]$. We can construct a rectilinear planar representation $H_{\mu_l}$ of $G_{\mu_l}$ with spirality $\sigma_{\mu_l}=\sigma_\nu+1$ and a rectilinear planar representation $H_{\mu_r}$ of $G_{\mu_r}$ with spirality $\sigma_{\mu_r}=\sigma_\nu-2$. Notice that, for this choice, we have $\sigma_{\mu_l}-\sigma_{\mu_r}=3$, hence, there exists a rectilinear planar representation $H_\nu$ of $G_\nu$ given the values of $\sigma_{\mu_l}$ and $\sigma_{\mu_r}$. We have to prove that the spirality of $H_\nu$ is $\sigma_\nu$.
	By Lemma~\ref{le:spirality-P-node-2-children} we have $\sigma_\nu' = \sigma_{\mu_l} - \frac{1}{2} \alpha_{u}^l - \frac{1}{2}\alpha_{v}^l$, where $\sigma_\nu'$ is the spirality of the representation $H_\nu$ of $G_\nu$ given a choice of $\sigma_{\mu_l}$, $\alpha_{v}^l$, and $\alpha_{u}^l$. Since $\sigma_\nu=\sigma_{\mu_l}-1$, $\alpha_{u}^l=1$, and $\alpha_{v}^l=1$, we have  $\sigma_\nu' = \sigma_\nu+1-\frac{1}{2}-\frac{1}{2}=\sigma_\nu$.
	
	\smallskip\noindent{\bf Case 2: $d=d'=r$.} This case is symmetric to Case 1.
	
	\smallskip\noindent{\bf Case 3: $d=l$, $d'=r$.}  In this case $G_\nu$ is of type \Pin{3lr} and we prove that $I_\nu = [\max\{m_l-1,m_r+2\}-\frac{1}{2},\min\{M_l-1, M_r+2\}-\frac{1}{2}]=[\max\{m_l-\frac{3}{2},m_r+\frac{3}{2}\},\min\{M_l-\frac{3}{2}, M_r+\frac{3}{2}\}]$.
	
	Assume first that $\sigma_\nu$ is the spirality of a rectilinear planar representation of $G_\nu$.  Since for a \Pin{3lr} component we have $k_u^r=k_v^l=\frac{1}{2}$ and $k_u^l=k_v^r=1$, by Lemma~\ref{le:spirality-P-node-2-children} we 
	have $\sigma_\nu = \sigma_{\mu_r} + \frac{1}{2}\alpha_{u}^r + \alpha_{v}^r$ and $\sigma_\nu =\sigma_{\mu_l} - \alpha_{u}^l - \frac{1}{2}\alpha_{v}^l$. Since $\alpha_u^l=\alpha_v^l=\alpha_u^r=\alpha_v^r=1$, we have:
	$\sigma_\nu \geq m_r + \frac{3}{2}$, $\sigma_\nu \leq M_r+\frac{3}{2}$, $\sigma_\nu \geq m_l-\frac{3}{2}$ and $\sigma_\nu \leq M_l-\frac{3}{2}$.
	
	We now show that, if $\sigma_\nu \in [\max\{m_l-\frac{3}{2},m_r+\frac{3}{2}\},\min\{M_l-\frac{3}{2},M_r+\frac{3}{2}\}]$, there exists a rectilinear planar representation of $G_\nu$ with spirality $\sigma_\nu$. 	We have $\sigma_\nu \in [m_l-\frac{3}{2}, M_l-\frac{3}{2}]$ and $\sigma_\nu \in [m_r+\frac{3}{2}, M_r+\frac{3}{2}]$. Hence, $\sigma_\nu+\frac{3}{2} \in [m_l, M_l]$ and $\sigma_\nu-\frac{3}{2} \in [m_r, M_r]$. We can construct a rectilinear planar representation $H_{\mu_l}$ of $G_{\mu_l}$ with spirality $\sigma_{\mu_l}=\sigma_\nu+\frac{3}{2}$ and a rectilinear planar representation $H_{\mu_r}$ of $G_{\mu_r}$ with spirality $\sigma_{\mu_r}=\sigma_\nu-\frac{3}{2}$. Notice that, for this choice, we have $\sigma_{\mu_l}-\sigma_{\mu_r}=3$, hence, there exists a rectilinear planar representation $H_\nu$ of $G_\nu$ given the values of $\sigma_{\mu_l}$ and $\sigma_{\mu_r}$. We have to prove that the spirality of $H_\nu$ is $\sigma_\nu$.
	By Lemma~\ref{le:spirality-P-node-2-children} we have $\sigma_\nu' = \sigma_{\mu_l} - \alpha_{u}^l - \frac{1}{2}\alpha_{v}^l$, where $\sigma_\nu'$ is the spirality of the representation $H_\nu$ of $G_\nu$ given a choice of $\sigma_{\mu_l}$, $\alpha_{v}^l$, and $\alpha_{u}^l$. Since $\sigma_\nu=\sigma_{\mu_l}-\frac{3}{2}$, $\alpha_{u}^l=1$, and $\alpha_{v}^l=1$, we have $\sigma_\nu' = \sigma_\nu+\frac{3}{2}-1-\frac{1}{2}=\sigma_\nu$.
\end{proof}


\end{document}